\documentclass[final]{IEEEtran}

\usepackage{amsmath,amssymb,amsfonts,mathrsfs,bm,amsthm}
\usepackage{amstext}
\usepackage{textcomp}
\usepackage{upgreek}
\usepackage{multicol}
\usepackage{indentfirst}
\usepackage{graphicx}
\usepackage{paralist}
\usepackage{multirow}
\usepackage{tabularx}
\usepackage{floatrow}
\usepackage{cite}
\usepackage{citesort}
\usepackage{booktabs}
\usepackage{subfigure}

\newtheorem{corollary}{\bf Corollary}
\newtheorem{proposition}{\bf Proposition}

\newtheorem{lemma}{\bf Lemma}
\newtheorem{remark}{\bf Remark}

\usepackage{accents}
\newcommand{\ubar}[1]{\underaccent{\bar}{#1}}

\IEEEoverridecommandlockouts
\allowdisplaybreaks
\begin{document}
\title{ Caching Policy Toward  Maximal Success Probability and Area Spectral Efficiency of Cache-enabled HetNets}
\author{\normalsize Dong Liu and Chenyang Yang
\thanks{This work was supported in part by National Natural Science Foundation of China (NSFC) under Grants with No. 61671036 and 61429101. The work of D. Liu was supported by the Academic Excellence Foundation of BUAA for Ph.D. Students. Parts of this work were presented in IEEE ICC 2016 \cite{DongICC} and IEEE GLOBECOM 2016 \cite{DongGC}. 
	
D. Liu and C. Yang are with School of Electronics and Information Engineering, Beihang University, Beijing, China	(e-mail: \{dliu, cyyang\}@buaa.edu.cn). }
}

\maketitle

\begin{abstract}
In this paper, we investigate the optimal caching policy respectively maximizing the success probability and area spectral efficiency (ASE) in a cache-enabled heterogeneous network (HetNet) where a tier of multi-antenna macro base stations (MBSs) is overlaid with a tier of helpers with caches. Under the probabilistic caching framework, we resort to stochastic geometry theory to derive the success probability and ASE. After finding the optimal caching policies, we analyze the impact of critical system parameters and compare the ASE with traditional HetNet where the MBS tier is overlaid by a tier of pico BSs (PBSs) with limited-capacity backhaul. Analytical and numerical results  show that the optimal caching probability is less skewed among helpers to maximize the success probability when the ratios of MBS-to-helper density, MBS-to-helper transmit power, user-to-helper density, or the rate requirement are small, but is more skewed to maximize the ASE in general. Compared with traditional HetNet, the helper density is much lower than the PBS density to achieve the same target ASE. The helper density can be reduced by increasing cache size. With given total cache size within an area, there exists an optimal helper node density that maximizes the ASE.
\end{abstract}


\section{Introduction}
To support the 1000-fold higher throughput in the fifth-generation (5G)
cellular systems, a promising way is to densify the network
by deploying more small base stations (BSs) in a macro cell \cite{densification}. Such
heterogeneous networks (HetNets) can  offload  the traffic and increase the area spectral efficiency (ASE), but the gain largely relies on
high-speed backhaul links. Although optical
fiber can provide high capacity, bringing fiber-connection to every single
small BS is rather labor-intensive and expensive. Alternatively, digital
subscriber line (DSL) or microwave backhaul may easily become a bottleneck
and frustratingly impair the throughput gain brought by the network densification
\cite{Femtocell}.

Recently, it has been observed that a large portion of mobile data
traffic is generated by many duplicate downloads of a few popular contents
\cite{woo2013comparison}. Besides, the storage capacity of today's
memory devices grows rapidly at a relatively low cost. Motivated by these facts, the authors  in \cite{Andy2012}
suggested to replace small BSs by the BSs  that have weak backhaul links (or even completely without backhaul)
but have high capacity caches, called \emph{helper nodes}. By optimizing the caching policies to serve
more users under the constraints of file downloading time, large throughput
gain was reported. Considering small cell networks (SCNs) with backhaul of
very limited capacity, the
authors in \cite{Procach14} observed that the backhaul traffic load can be
reduced by caching files at the small BSs  based on their popularity. These results indicate that by fetching
contents locally instead of fetching from core network via backhaul links duplicately,
equipping caches at BSs is a promising way to unleash the
potential of HetNet.

In \cite{Dong}, both the throughput and energy efficiency of cached-enabled cellular networks with hexagonal cells were analyzed. For HetNets or SCNs, however, it is more appropriate to use Poisson Point Process (PPP) to model the BS location \cite{flexible}. Stochastic geometry theory was first applied in \cite{EURASIP} for a single-tier cache-enable SCN, where each SBS caches the most popular files and user density is assumed sufficiently large such that all the BSs are active. In \cite{BastugISIT}, the results in \cite{EURASIP} are extended  to a two-tier network with intra-tier and inter-tier dependence. In \cite{chenchenyang}, the throughput of a cache-enabled network with content pushing to users, device-to-device communication and caching at relays was derived, where every node (including the macro BS (MBS) and relay) is with a single antenna and with high-capacity backhaul.

Caching policy is critical in reaping the benefit brought by caching \cite{DongMag}. In wireless networks, when the coverage of several BSs overlaps, a user is able to fetch contents from multiple helpers and hence cache-hit probability can be increased by caching different files among helpers \cite{Niki13}. However, owing to interference and path loss, such a file diversity may lead to low signal-to-interference-plus-noise ratio (SINR), since a user may associate with a relative further BS to ``hit the cache" when the nearest BS does not cache the requested file \cite{Dong}. This suggests that caching policy not only affects the cache-hit probability but also changes the way of user association and hence the SINR distribution of wireless networks.

In \cite{Niki13}, caching policy was optimized to minimize the file download time where the interference among helpers are not considered. In \cite{song2015optimal}, the optimal caching policy was proposed to minimize the average bit error rate over fading channels. Both \cite{Niki13} and \cite{song2015optimal} assume \emph{a priori} known BS-user topology, which is not practical in mobile networks.  To reflect the uncertain connectivity between BS and user, a probabilistic caching framework was proposed recently in \cite{Blaszczyszyn2015optimal}, where the network models based on stochastic geometry were considered. In particular, a probabilistic caching policy was proposed where each BS caches files independently according to an optimized caching probability to maximize  the cache-hit probability. However, the optimal caching probability is not obtained with closed-form, which makes it hard to gain useful insights into the impact of various system parameters. In \cite{cui2015analysis}, the optimal caching probability maximizing the successful transmission probability in a single-tier network was obtained in closed-form when user density approaches infinity. In \cite{rao2015optimal}, caching policy was optimized to maximize the traffic offloaded to helpers and cache-enabled users, but the links among helpers and users are assumed interference-free. In real-world HetNets, the interference is complicated and has large impact on the system design and network performance. While deploying helpers is a cost-effective way for offloading traffic and increasing ASE of cellular networks, how to place the contents and how to deploy the helpers for interference-limited cache-enabled HetNets are still not well understood. Specifically:

\begin{itemize}
	\item Existing works rarely consider caching policy optimization in HetNets (e.g., only  single-tier network is considered in \cite{EURASIP,cui2015analysis}, and  each BSs caches the most popular files in \cite{EURASIP,BastugISIT,chenchenyang}), and the basic features of cache-enabled HetNets, such as cross-tier interference, the transmit power and BS density of different tiers, user densities, rate requirement and association bias, are not well captured and analyzed in \cite{Niki13,BastugISIT,rao2015optimal,Blaszczyszyn2015optimal,song2015optimal,EURASIP,BastugISIT,chenchenyang,cui2015analysis}.
	\item Existing works mainly focus on finding the optimal caching probability maximizing the success probability in cache-enabled networks \cite{Blaszczyszyn2015optimal,cui2015analysis,rao2015optimal}, but never consider caching policy optimization for maximizing the ASE, yet another important performance metric.
	\item Existing works never consider helper idling \cite{BastugISIT,EURASIP,chenchenyang,Niki13,song2015optimal,Blaszczyszyn2015optimal,cui2015analysis,rao2015optimal} (i.e., turning helpers with no user to serve into idle mode to avoid generating interference), which is appealing since the cost-effective helpers make dense deployment possible, and affects both SINR distribution and caching policy optimization.
\end{itemize}

To better understand these issues, in this paper we attempt to find the optimal caching policy for cache-enabled HetNet under different performance metrics, analyze the impact of critical system parameters, investigate the benefits of cache-enabled HetNet with respect to traditional HetNet with limited-capacity backhaul PBSs, and reveal the tradeoff in deploying cache-enabled HetNets. To this end, we consider a cache-enabled HetNet where a tier of multi-antenna MBSs with high capacity backhaul is overlaid with a tier of denser helpers with caches but without backhaul links. Under the probabilistic caching framework \cite{Blaszczyszyn2015optimal} and using stochastic geometry theory, we derive the success probability and ASE of the cache-enabled HetNets respectively as functions of
BS/helper node density, user density, caching probability, storage size and file popularity. We also derive the ASE of traditional HetNet
where a tier of MBSs is overlaid with a tier of PBS with limited-capacity backhaul for
comparison. The major contributions of this paper are summarized as follows,
\begin{itemize}
	\item We find the optimal caching policy to maximize a concave lower bound of the success probability, which is very tight in low user density case. Analytical and simulation results show that the optimal caching probability is less skewed to maximize the success probability when the ratios of MBS-to-helper density, MBS-to-helper transmit power, user-to-helper density, or the rate requirement are low, and with helper idling.
	\item We find the optimal caching policy to maximize the approximated ASE, which is accurate for large cache size of each helper. When the user-to-helper density and helper-to-MBS density approach infinity, we prove that simply caching the most popular files everywhere is the optimal solution, which is quite different from maximizing the success probability.
	\item We show that the cache-enabled HetNet can provide much higher ASE than the traditional HetNet with the same PBS/helper density, or can be deployed with much lower helper density than the PBS density of traditional HetNet to achieve the same target ASE, which can reduce the cost of deployment and operation remarkably. Moreover, we find that the helper density can be traded off by the cache size to achieve a target ASE. Given the total cache size within an area, there exists an optimal helper density that maximizes the ASE.
\end{itemize}

The rest of this paper is organized as follows. In Section II, we present the system model. The caching policies toward maximizing success probability and ASE are optimized in Section III and Section IV, respectively. The numerical and simulation results are provided in Section V, and the conclusions are drawn in Section VI.

\section{System Model}

\subsection{Network Model}
We consider a cache-enabled HetNet, where a tier of MBSs is overlaid with a tier of denser helper nodes  as shown in Fig.~\ref{fig:layout}. Since MBSs are deployed with relative low density, we assume each MBS is connected to the core network with high-capacity backhaul links, e.g., optical fibers, and the helper nodes are deployed without backhaul but equipped with caches.

\begin{figure}[!htb]
\centering
\includegraphics[width=0.77\textwidth]{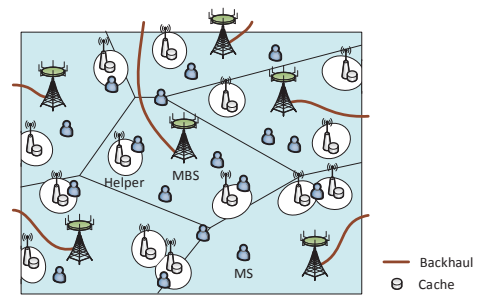}
\caption{Layouts of the considered cache-enabled HetNets.} \label{fig:layout}
\end{figure}

The distribution of MBSs, helper nodes and users are modeled as three independent homogeneous PPPs  with density of $\lambda_1$, $\lambda_2$ and $\lambda_u$, denoted as $\Phi_{1}$, $\Phi_{2}$ and $\Phi_u$, respectively. Each MBS is equipped with $M_1 \geq 1$ antennas and each helper node is with $M_2 = 1$ antenna.\footnote{Since helper nodes are expected to enable dense deployment with low cost, we only consider single-antenna case.} Denote $k\in\{1,2\}$ as the index of the tier that a randomly chosen user in the network (called the typical user) is associated with. In the following, if not specified, BS refers to either MBS or helper. The transmit power at each BS in the $k$th tier is denoted by $P_k$.

We assume that each user requests a file from a content catalog that contains $N_f$ files randomly. The files are indexed according to their popularity, ranking from the most popular (the $1$st file) to the least popular (the $N_f$th file). The probability that the $f$th popular file is requested follows Zipf distribution as
\begin{equation}
p_f = \frac{f^{-\delta}}{\sum_{n=1}^{N_f} n^{-\delta}}, \label{eqn:pf}
\end{equation}
where the skew parameter $\delta$ is with typical value of 0.5 $\sim$ 1.0 \cite{breslau1999web}. For simplicity, we assume that the
files are with equal size\footnote{Files with different size can be divided into equal-size content chunks.}
and the cache size of each helper node is $N_c$.

We consider probabilistic caching policy where each helper independently selects files to cache according to a specific probability distribution. To unify the analysis, denote $0 \leq q_{f, k} \leq 1$ as the probability that the BS of the $k$th tier caches the $f$th file. When $\mathbf q_{k} \triangleq [q_{f, k}]_{f= 1,\cdots, N_f}$ is given, each BS can determine which files should be cached by the method in \cite{Blaszczyszyn2015optimal}. For the MBS tier,  $\mathbf q_{1} = \mathbf 1$, since the MBS can be regarded as caching all the files due to the high-capacity backhaul.


Since every helper caches files independently, the distribution of the BSs in the $k$th tier caching the $f$th file can be regarded as a thinning of the PPP $\Phi_k$ with probability $q_{f,k}$, which follows a PPP with density $q_{f, k}\lambda_k$ (denoted by $\Phi_{f,k}(\mathbf{q}_1,\mathbf{q}_2)$). Similarly, the distribution of the BSs in the $k$th tier not caching the $f$th file follows PPP with density $(1-q_{f, k})\lambda_k$ (denoted by $\Phi_{f',k}(\mathbf{q}_1,\mathbf{q}_2)$).

We consider user association based on both channel condition and cached files in each helper. Specifically, when the typical user requests the $f$th file, it associates with the BS in the set of $\{\Phi_{f,k}(\mathbf{q}_1,\mathbf{q}_2)\}_{k=1,2}$ that has the strongest average biased-received-power (BRP). The BRP for the $k$th tier is $P_{{\rm r}, k} = P_kB_k r^{-\alpha}$, where $B_k \geq 0$ is the association bias factor, $r$ is the BS-user distance, and $\alpha$ is the path-loss exponent. In such a cache-enabled HetNet, the user may not associate with the BS with the strongest BRP since the BS may not cache the requested file, which is very different from the traditional HetNets without local caching.

 Since the helpers  without backhaul can be densely deployed at low cost and the traffic may fluctuate among peak and off-peak times, the density of helper nodes may become comparable with or even higher than the density of users and hence some helpers may have no users to serve. These ``inactive BSs" will be turned into idle mode (i.e., become muting) to avoid generating interference, and the rest BSs that have users to serve are called as ``active BSs" denoted by $\tilde{\Phi}_{k}$. We assume that $\lambda_u \gg \lambda_1$, and the number of users in each macro cell far exceeds the number of antennas at each MBS $M_1$. Assume that each MBS serves every $M_1$ users in the same time-frequency resource  by zero-forcing beamforming with equal power allocation, while each helper serves its associated users by time division or frequency division multiple access with full power. These assumptions define a simplified scenario, which however can capture the fundamental features of cache-enabled HetNets.
Then, the downlink SINR at the typical user that requests the $f$th file and associates with the $k$th tier is
\begin{align}
&{\gamma}_{f,k}(\mathbf{q}_1, \mathbf{q}_2)
= \frac{\frac{P_k}{M_k} h_{k0} r_{k}^{-\alpha}}{\sum_{j=1}^{2} ( I_{f,kj} + I_{f',kj})  + \sigma^2 }
\triangleq \frac{\frac{P_k}{M_k} h_{k0} r_{k}^{-\alpha}}{I_{k}+ \sigma^2 }, \label{eqn:gamma}
\end{align}
where $h_{k0}$ is the equivalent channel power (including channel coefficient and beamforming) from the associated BS $b_{k0}$ to the typical user, $r_k$ is the corresponding distance, $\sigma^2$ is the noise power, and the total interference $I_k$ consists of the interference from the BSs that cache the $f$th file, i.e., $I_{f, kj} = \sum_{i \in \tilde{\Phi}_{f,j}(\mathbf{q}_1,\mathbf{q}_2) \backslash  b_{k0}} P_jh_{ji} r_{ji}^{-\alpha}$, and the interference from the BSs that do not cache the $f$th file, i.e., $I_{f', kj} = \sum_{i \in \tilde{\Phi}_{f',j}(\mathbf{q}_1,\mathbf{q}_2)}  P_jh_{ji} r_{ji}^{-\alpha}$, where $\tilde \Phi_{f, j}(\mathbf{q}_1,\mathbf{q}_2)$ and $\tilde \Phi_{f', j}(\mathbf{q}_1,\mathbf{q}_2)$ are respectively the sets of active BSs  in the $j$th tier that caching and not caching the $f$th file, $h_{ji}$ and $r_{ji}$ are the equivalent interference channel power and distance from the $i$th active BS in the $j$th tier to the typical user. We consider Rayleigh fading channels. Then, $h_{k0}$ follows exponential distribution with unit mean, i.e., $h_{k0} \sim \exp (1)$,\footnote{For $k = 1$, since each MBS serves $M_1$ users on the same time-frequency resource, $h_{10} \sim \mathbb{G}(M_1 - M_1 + 1, 1)$\cite{zhang2011multi}, which is actually exponential distribution with unit mean. For $k = 2$, since each helper has one antenna, $h_{20}\sim \exp(1)$.} and $h_{ij}$ follows gamma distribution with shape parameter $M_j$ and unit mean, i.e., $h_{ji} \sim \mathbb{G}(M_j, 1/M_j)$ \cite{adhoc}.  Since $\mathbf{q}_1 = \mathbf{1}$ is a constant due to the high capacity backhaul of MBS, we denote any functions of $(\mathbf{q}_1, \mathbf{q}_2)$ as $(\mathbf{q}_2)$ for notational simplicity in the following, e.g., $\gamma_{f,k}(\mathbf{q}_1,\mathbf{q}_2) = \gamma_{f,k}(\mathbf{q}_2)$.

Since HetNets are usually interference-limited \cite{flexible}, it is reasonable to neglect the thermal noise, i.e., $\sigma^2 = 0$. For notational simplicity, we define the ratios of BS density, number of antennas, transmit power, and bias factor as $\lambda_{jk} \triangleq \lambda_j/\lambda_k$, $M_{jk} \triangleq M_j/M_k$, $P_{jk} \triangleq P_j/P_k$, and $B_{jk} \triangleq B_j/B_k$, respectively. Note that $\lambda_{kk} = M_{kk} = P_{kk} = B_{kk} = 1$.

\subsection{Performance Metric}
We consider two performance metrics, either from each user or from network perspective.

From the user perspective, we use success probability to reflect the quality of service, which is the probability that the achievable rate of the typical user exceeds the rate requirement $R_0$. Based on the law of total
probability, the success probability is given by
\begin{align}
p_{\rm s}(\mathbf{q}_2)  = & \sum_{k = 1}^{2}  p_{{\rm s}, k}(\mathbf{q}_2) = \sum_{k = 1}^{2} \sum_{f=1}^{N_f} p_f \mathcal{P}_{f,k}(\mathbf{q}_2)\nonumber\\
& \times \mathbb{P} \left( \frac{W}{ U_{f,k}/M_k }\log_2 (1 + \gamma_{f,k}(\mathbf{q}_2)) \geq R_{0} \right)  , \label{eqn:eql} 
\end{align}
where $p_{{\rm s}, k}(\mathbf{q}_2)$ is the the success probability contributed by the $k$th tier, $\mathcal{P}_{f,k}(\mathbf{q}_2)$ is the probability of the typical user associated with the $k$th tier when requesting the $f$th file which is given by Lemma 1 in the next section, $\mathbb{P} \left( \frac{W}{  U_{f,k}/M_k }\log_2 (1 + \gamma_{f,k}(\mathbf{q}_2)) \geq R_{0} \right)$ is the success probability when the user requests the $f$th file and associates with the $k$th tier, $W$ is the total transmission bandwidth of the network, $U_{f,k}$ is the total number of users served by the same BS together with the typical user (including the typical user) when the typical user requests the $f$th file. Since every $M_k$ users can share the same time-frequency resource thanks to zero-forcing beamforming, the total time-frequency resource is divided into $\frac{U_{f,k}}{M_k}$ parts and hence $\frac{W}{U_{f,k}/M_k}$ represents the time-frequency resource allocated to each user.

From the network perspective, we use ASE to measure the network capacity, which is defined as the average throughput per unit area normalized by the transmission bandwidth \cite{Quek},
\begin{equation}
{\sf ASE}(\mathbf{q}_2) = \frac{1}{W}\sum_{k=1}^{2} p_{{\rm a},k}(\mathbf{q}_2)\lambda_k  \mathbb{E}[R_k(\mathbf{q}_2)] , \label{eqn:ASEdef}
\end{equation}
where $p_{{\rm a},k}(\mathbf{q}_2)$ is the probability that a randomly chosen BS in the $k$th tier is active, $R_k(\mathbf{q}_2) \triangleq \sum_{u=1}^{U_k}\frac{W}{  U_k/M_k }$ $\log_2 (1 + \gamma_{uk}(\mathbf{q}_2))$ is the throughput of the active BS in the $k$th tier, $U_k$ is the total number of users served by a randomly chosen active BS in the $k$th tier, $\gamma_{uk}(\mathbf{q}_2)$ is the SINR of the $u$th user served by the BS in the $k$th tier, and $\mathbb{E}$ denotes the expectation. The expectation is taken over file requests, cached files among helpers, the number and location of BSs and users, and small-scale fading.

\section{Caching Policy Maximizing Success probability}
In this section, we first derive the success probability as a function of caching probability, then
find the optimal caching probability maximizing the success probability and analyze the impact of different system settings on the optimal caching policy.

\subsection{Successful Probability}
To derive the main results, we first introduce a lemma regarding the tier association probability.

\begin{lemma}
	The probability of the typical user associating with the $k$th tier is
	\begin{equation}
	\mathcal{P}_{k}(\mathbf{q}_2) =  \sum_{f = 1}^{N_f} \frac{p_fq_{f,k}}{\sum_{j=1}^{2}q_{f,j}\lambda_{jk} (P_{jk}B_{jk})^{\frac{2}{\alpha}}}, \label{eqn:asso}
	\end{equation}
	and the probability of the typical user associating with the $k$th tier conditioned on that the user requests the $f$th file is $\mathcal{P}_{f,k}( \mathbf{q}_2)  = q_{f,k}\left(\sum_{j=1}^{2}q_{f,j}\lambda_{jk} (P_{jk}B_{jk})^{\frac{2}{\alpha}} \right)^{-1}$.
\end{lemma}
\begin{IEEEproof}
See Appendix A.
\end{IEEEproof}

From Lemma 1, we can see that both the requested file and caching probability affect user association, which is very different from the traditional HetNets without local caching.

For the tractability of the analysis, the following approximations are used to derive the closed-form expression of success probability.

{\bf BS Active Probability Approximation}:
The distribution of the active BSs in $k$th tier is modeled as PPP with density $p_{{\rm a},k}(\mathbf{q}_2)\lambda_k$ by thinning ${\Phi}_k$ with probability $p_{{\rm a},k}(\mathbf{q}_2)$ \cite{economy}, and $p_{{\rm a},k}(\mathbf{q}_2)$ can be approximated as
\begin{align}
p_{{\rm a},k}(\mathbf{q}_2) & = 1 - \int_{0}^{\infty} e^{-\lambda_u x} f_{S_k}(x){\rm d} x \nonumber \\
& \approx 1 - \left(1 + \frac{\mathcal{P}_k(\mathbf{q}_2)\lambda_u}{3.5\lambda_k}\right)^{-3.5}, \label{eqn:pa}
\end{align}
where $f_{S_k}(x)$ is the PDF of the service area  $S_k$ of the BS in the $k$th tier, and the approximation comes from $f_{S_k}(x) \approx \frac{3.5^{3.5}}{\Gamma(3.5)} \big(\tfrac{\lambda_j}{\mathcal{P}_k(\mathbf{q}_2)}\big){}^{3.5}x^{2.5}e^{-3.5{\lambda_kx}/{\mathcal{P}_k}}$ in \cite{offloading}, which is accurate when $\lambda_2 \gg \lambda_1$. Note that for $k=1$ (i.e., the MBS tier), $p_{{\rm a},1}(\mathbf{q}_2)$ exactly equals to 1 since ${\lambda_u}/{\lambda_1} \to \infty$.

{\bf Average Load Approximation}: The number of users served by the same BS together with the typical user when the typical user requests the $f$th file is approximated by the average value of $U_k$, i.e., $U_{f,k} \approx \mathbb{E}[U_k] = 1 + \frac{1.28\lambda_u \mathcal{P}_{k}(\mathbf{q}_2)}{\lambda_k} \label{eqn:Uk} $ \cite{offloading}, which is verified to have negligible impact on the success probability by simulation in Section V. Note that when ${\lambda_u}/{\lambda_k}\to 0$, i.e., the user density is low such that each helper serves at most one user, then $U_{f,k} \to \mathbb{E}[U_k] = 1$, and the approximation becomes exact.

Then, based on the tier association probability and the above approximations, we obtain the success probability in the following proposition.

\begin{proposition}
	The success probability of the typical user is
	\begin{align}
	 p_{\rm s}(\mathbf{q}_2)  = &\sum_{k=1}^{2}\sum_{f=1}^{N_f} p_f  q_{f,k} \Bigg(\sum_{j=1}^{2}  \lambda_{jk} P_{jk}^{\frac{2}{\alpha}} \nonumber \\
	& \times \left(q_{f,j}  p_{{\rm a},j}(\mathbf{q}_2)) B_{jk}^{\frac{2}{\alpha}} \mathsf{Z}_{1,jk}(\gamma_{0,k}(\mathbf{q}_2)) \right. \nonumber\\
	&\left.\left.+ (1\!-\!q_{f,j})  p_{{\rm a},j}(\mathbf{q}_2))\mathsf{Z}_{2,jk}(\gamma_{0,k}(\mathbf{q}_2))  +  q_{f,j}   B_{jk}^{\frac{2}{\alpha}} \right) \vphantom{\sum_{j=1}^{2}}\right)^{-1}\!\!\!\!, \label{eqn:suc0} 	
	\end{align}
		where $\gamma_{0,k}(\mathbf{q}_2) = 2^{\frac{R_0}{WM_k}  \big( 1 + \frac{1.28\lambda_u \mathcal{P}_{k}(\mathbf{q}_2)}{\lambda_k}\big) } - 1$ is the equivalent per-user receive SINR requirement when the user associates with the $k$th tier, $ \mathsf{Z}_{1,kj} (x)\triangleq {}_{2}F_1 \big[ -\frac{2}{\alpha}, M_j; 1-\frac{2}{\alpha}; -\frac{x}{M_{jk}B_{jk}} \big] -1$, $\mathsf{Z}_{2,kj}(x) \triangleq \Gamma \left(1-\frac{2}{\alpha}\right)  \Gamma\left(M_j + \frac{2}{\alpha}\right) \Gamma(M_j)^{-1}(\frac{x}{M_{jk}})^{\frac{2}{\alpha}}$, $_{2}F_1[\cdot]$ and $\Gamma(\cdot)$ denote the Gauss hypergeometric function and Gamma function, respectively.
\end{proposition}
\begin{IEEEproof}
	See Appendix B.
\end{IEEEproof}

Specifically, by substituting $k=2$, $q_{f,1} = 1$, $M_2 = 1$ and $p_{{\rm a}, 1}(\mathbf{q}_1) = 1$ into \eqref{eqn:suc0}, we  obtain the success probability contributed by the helper tier as \eqref{eqn:Poff},
\begin{figure*}[ht]
	\begin{equation}
	p_{{\rm s}, 2} (\mathbf{q}_2) = \sum_{f=1}^{N_f} \frac{ p_f q_{f,2}}{ \mathsf{C}_{1}(\gamma_{0,2}(\mathbf{q}_2)) + \mathsf{C}_{2}(\gamma_{0,2}(\mathbf{q}_2)) p_{{\rm a},2}(\mathbf{q}_2) +  \mathsf{C}_{3}(\gamma_{0,2}(\mathbf{q}_2)) p_{{\rm a},2}(\mathbf{q}_2) q_{f,2}  + q_{f,2} }, \label{eqn:Poff} 
	\end{equation}
	\hrulefill
\end{figure*}
where $\mathsf{C}_{1}(x) \triangleq \lambda_{12}  (P_{12}B_{12})^{\frac{2}{\alpha}} {}_{2}F_1 \big[ -\frac{2}{\alpha}, M_1; 1-\frac{2}{\alpha}; -\frac{x}{M_{12}B_{12}} \big]$, $\mathsf{C}_{2}(x) \triangleq  \Gamma (1-\frac{2}{\alpha}) \Gamma(M_2 + \frac{2}{\alpha_2}) \Gamma(M_2)^{-1}x^{\frac{2}{\alpha}}$, $\mathsf{C}_{3}(x) \triangleq {}_{2}F_1  \big[ -\frac{2}{\alpha}, M_2; 1-\frac{2}{\alpha}; -x \big] - \mathsf{C}_{2}(x)- 1$.

It is shown that $p_{{\rm s}, 2}(\mathbf q_{2})$ decreases with the helper active probability $p_{{\rm a}, 2}(\mathbf{q}_2)$. This indicates that if a helper can be turn into idle mode when no user associates with it (i.e., $p_{{\rm a}, 2}(\mathbf{q}_2) < 1$), the success probability is higher than that without helper idling (i.e., $p_{{\rm a}, 2}(\mathbf{q}_2) = 1$). Proposition 1 also shows that $p_{{\rm a}, 2}(\mathbf{q}_2)$ and $\gamma_{0,2}(\mathbf{q}_2)$ increase with user density $\lambda_u$. Further considering that $p_{{\rm s}, 2}(\mathbf q_{2})$ decreases with $\gamma_{0,2}(\mathbf{q}_2)$, $p_{{\rm s}, 2}(\mathbf q_{2})$ decreases with user density.

\subsection{Optimal Caching Policy}
The optimal caching probability that maximizes the success probability can be found from
\begin{subequations}
	\begin{align}
	\text{\em Problem 1:}\quad \max_{\mathbf q_2}~ & p_{\rm s}(\mathbf q_{2})  \nonumber \\
	{\rm s.t.}~ & \sum_{i=1}^{N_f} q_{f,2} \leq N_c  \label{eqn:con1}  \\
	& 0\leq q_{f,2} \leq 1, ~ f = 1, \cdots, N_f \label{eqn:con2}
	\end{align}
\end{subequations}
where \eqref{eqn:con1} is equivalent to the cache size constraint (i.e., the number of cached file cannot exceed the cache size) for each helper as proved in \cite{Blaszczyszyn2015optimal}, and \eqref{eqn:con2} is the probability constraint.

This problem is not concave in general, because both $p_{{\rm a}, 2}(\mathbf{q}_2)$ and $\gamma_{0,2}(\mathbf{q}_2)$ in the objective function depend on $\mathcal{P}_2(\mathbf{q}_2)$, which is a complicated function as shown in \eqref{eqn:asso}. Moreover, the summation over two tiers make the problem more complicated. Therefore, only a local optimal solution can be found, say by using interior point method \cite{boyd2004convex}. Such a solution is not only of high complexity when the dimension of optimization variable $\mathbf q_2$, i.e., $N_f$, is large, but also hard to provide design guideline for practical systems.

From \eqref{eqn:eql}, we can see that the success probability contributed by each tier is the product of tier association probability and the probability of tier success probability. Since we assume that $\lambda_u \gg \lambda_1$ and $M_1$ is not large, when user associates with the MBS tier, the time-frequency resources allocated to the user could be too stringent to support the rate requirement. On the other hand, the probability of the user associated with the MBS tier is low when helper nodes are densely deployed, i.e., $\lambda_2 \gg \lambda_1$. Therefore, we can safely neglect the impact of the success probability contributed by the MBS tier $p_{{\rm s}, 1}(\mathbf{q}_2)$,  and maximizing the success probability contributed by the helper tier $p_{{\rm s}, 2}(\mathbf{q}_2)$ as,\footnote{In Section V, we show that for rate requirement $R_0 > 1$ Mbps, when $\lambda_u/\lambda_1 = 50$ and $\lambda_2/\lambda_1 = 50$, the success probability contributed by the first tier $p_{{\rm s}, 1}(\mathbf{q}_2)< 0.05$, which can be safely neglected. We also show that maximizing $p_{{\rm s},2}(\mathbf{q}_2)$ can achieve almost the same performance as maximizing $p_{{\rm s}}(\mathbf{q}_2)$ in more general cases.}
\begin{subequations}
	\begin{align}
	\text{\em Problem 2:}\quad \max_{\mathbf q_2}~ & p_{{\rm s}, 2}(\mathbf q_{2})  \nonumber \\
	{\rm s.t.}~ & \eqref{eqn:con1}, \eqref{eqn:con2}  \nonumber
	\end{align}
\end{subequations}

In the following we first solve Problem 2 in general case, and then analyze the behavior of optimal caching policy and the impact of system settings in high and low user density cases.

\subsubsection{General Case}
Problem 2 is still not concave in general since the objective function still related to $p_{{\rm a}, 2}(\mathbf{q}_2)$ and $\gamma_{0,2}(\mathbf{q}_2)$.
Considering that it is the complicated expression of $\mathcal{P}_2(\mathbf{q}_2)$ that makes the problem non-concave, in the following, we first introduce a $\mathbf q_2$-independent upper bound for  $\mathcal{P}_2(\mathbf{q}_2)$, which yields an $\mathbf q_2$-independent upper bound for $p_{{\rm a}, 2}(\mathbf{q}_2)$ (denoted by $\bar p_{{\rm a}, 2}$) and $\gamma_{0,2}(\mathbf{q}_2)$ (denoted by $\bar{\gamma}_{0,2}$) respectively because $p_{{\rm a}, 2}(\mathbf{q}_2)$ and $\gamma_{0,2}(\mathbf{q}_2)$ increase with $\mathcal{P}_2(\mathbf{q}_2)$ as shown in \eqref{eqn:pa} and Proposition 1, respectively. Since $p_{{\rm s}, 2}(\mathbf{q}_2)$ decreases with $p_{{\rm a}, 2}(\mathbf{q}_2)$ as shown in \eqref{eqn:Poff} and decreases with the equivalent per-user receive SINR requirement $\gamma_{0,2}(\mathbf{q}_2)$, by substituting $\bar p_{{\rm a}, 2}$ and $\bar{\gamma}_{0,2}$ into \eqref{eqn:Poff}, we can obtain the lower bound of $p_{{\rm s}, 2}(\mathbf{q}_2)$ (denoted as $\ubar p_{{\rm s}, 2}(\mathbf{q}_2)$), which is shown to be a concave function of $\mathbf{q}_2$ and can be solved in closed-form.


Denote the caching probability that maximizes $\mathcal{P}_2(\mathbf{q}_2)$ as $\mathbf q_2^o \triangleq [ q_{f,2}^o]_{f = 1, \cdots, N_f}$. Then, for any $\mathbf q_{2}$ satisfying \eqref{eqn:con1} and \eqref{eqn:con2}, we have $\mathcal{P}_2 (\mathbf q_2) \leq \mathcal{P}_2(\mathbf q_2^o)$. Therefore, $\mathcal{P}_2(\mathbf q_2^o)$ can be used as the $\mathbf{q}_2$-independent upper bound of $\mathcal{P}_2(\mathbf{q}_2)$, which is tight when $\mathbf{q}_2 = \mathbf{q}_2^o$ and can be obtained by solving the following concave problem
\begin{align}
\text{\em Problem 3:} \max_{\mathbf q_2}~ & \mathcal{P}_2(\mathbf{q}_2) = \sum_{f=1}^{N_f}  \frac{p_f q_{f,2}}{ \lambda_{12} (P_{12} B_{12}){}^{\frac{2}{\alpha}}+ q_{f,2} } \label{eqn:obj} \\
{\rm s.t.}~ & \eqref{eqn:con1},\eqref{eqn:con2}   \nonumber
\end{align}
where \eqref{eqn:obj} is obtained by substituting $\mathbf q_1 = \mathbf 1$ into \eqref{eqn:asso}.

It can be easily proved that the Hessian matrix of $\mathcal{P}_2(\mathbf{q}_2)$ is negative definite. Further considering that constraints \eqref{eqn:con1} and \eqref{eqn:con2} are linear, Problem 3 is concave. Then, from the Karush-Kuhn-Tucker (KKT) condition, we obtain the optimal solution of Problem 3 as
\begin{equation}
	q_{f,2}^o = \left[ \frac{\sqrt{\lambda_{12} (P_{12} B_{12})^{\frac{2}{\alpha}}} }{\sqrt{\mu}}\sqrt{p_f} - \lambda_{12} (P_{12} B_{12})^{\frac{2}{\alpha}}  \right]_0^1\!, \!  \label{eqn:qo}
\end{equation}
where $[x]_0^1 = \max\{\min\{x,1\},0\}$ denotes that $x$ is truncated by $0$ and $1$, and the Lagrange multiplier $\mu$ satisfying $\sum_{f=1}^{N_f}q_{f,2}^o = N_c$ can be efficiently found by bisection searching.

By substituting $\mathcal{P}_2(\mathbf{q}_2^o)$ into \eqref{eqn:Poff}, we can obtain the lower bound of success probability as
\begin{align}
&\ubar p_{{\rm s}, 2}(\mathbf q_2) = \nonumber\\
& \sum_{f=1}^{N_f} \frac{ p_f q_{f,2}}{ \mathsf{C}_{1}(\bar \gamma_{0,2}) + \bar{p}_{{\rm a},2}  \mathsf{C}_{2}(\bar \gamma_{0,2}) +  (\bar{p}_{{\rm a},2} \mathsf{C}_{3}(\bar\gamma_{0,2}) + 1) q_{f,2} }. \!\! \label{eqn:Pofflower}
\end{align}
where $\bar \gamma_{0,2} = 2^{\frac{R_0}{W}  \big( 1 + \frac{1.28\lambda_u \mathcal{P}_2(\mathbf q_2^o)}{\lambda_2}\big) } - 1$ and $\bar p_{{\rm a}, 2} =  1 - \big(1 + \tfrac{\mathcal{P}_2(\mathbf q_2^o)\lambda_u}{3.5\lambda_2}\big)^{-3.5}$.
\begin{proposition}
The optimal caching probability that maximizes the lower bound of the success probability $\ubar p_{{\rm s}, 2}(\mathbf q_2)$ is
\begin{align}
\underline{q}_{f,2}^*  = & \left[ \frac{\sqrt{ \mathsf{C}_{1}(\bar\gamma_{0,2}) +  \bar{p}_{{\rm a},2} \mathsf{C}_{2}(\bar\gamma_{0,2})} }{\sqrt{\nu} ( \bar{p}_{{\rm a},2} C_{3,\bar\gamma_{0,2}} + 1)} \sqrt{p_f}\right. \nonumber\\
&\left.-\frac{\mathsf{C}_{1}(\bar\gamma_{0,2}) +  \bar{p}_{{\rm a},2} \mathsf{C}_{2}(\bar\gamma_{0,2})}{ \bar{p}_{{\rm a},2} C_{3, \bar\gamma_{0,2}} + 1} \right]_0^1, \label{eqn:optlower}
\end{align}
where $\nu$ satisfying $\sum_{f=1}^{N_f}\ubar{q}_{f,2}^* = N_c$ can be found by bisection searching.
\end{proposition}
\begin{IEEEproof}
Since $\bar p_{{\rm a}, 2}$ and $\bar \gamma_{0,2}$ do not depend on $q_{f,2}$,  \eqref{eqn:Pofflower} has the same function structure as \eqref{eqn:obj} and hence is a concave function. Then, by using the KKT condition similar to solving Problem 3, we obtain \eqref{eqn:optlower}.
\end{IEEEproof}

As shown in \eqref{eqn:optlower}, $\ubar q_{f,2}^*$ is non-increasing with $f$ since $p_f$ decreases with $f$, which coincides with the intuition that the file with higher popularity should be cached with higher probability.
Moreover, since  $\ubar q_{f,2}^*$ is truncated by $0$ and $1$, there may exist some files of high popularity (e.g. $f = 1,\cdots, N_1, ~N_1 \leq N_f$) with caching probability of $1$ (i.e., cached everywhere), and some files of low popularity (e.g. $f = N_f - N_0 + 1, \cdots, N_f, ~N_0 \leq N_f - N_1 + 1$) with caching probability of $0$ (i.e., not cached at all helpers).
For $\ubar q_{f,2}^* \in (0,1)$, considering  \eqref{eqn:pf}, the relation between caching probability and file popularity rank in \eqref{eqn:optlower} obeys a shifted power law with exponent $-\delta/2$, which is very different from the noise-limited scenario considered in \cite{rao2015optimal}.

\begin{figure}[!htb]
\centering
\includegraphics[width=0.75\textwidth]{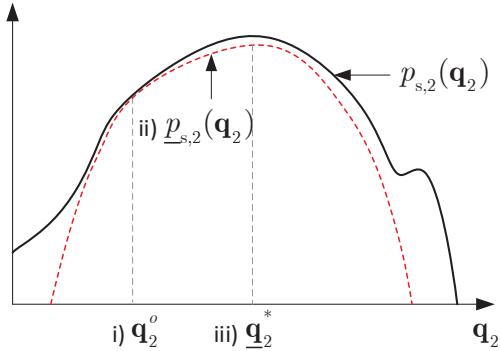}
\caption{Illustration the procedure of finding the closed-form expression of caching probability $ \mathbf{\underline{q}}_2^*$.  i) Solve Problem 3 to obtain $\mathbf q_2^0$ and $\mathcal{P}_2(\mathbf q_2^o)$. ii) Substitute $\mathcal{P}_2(\mathbf q_2^o)$ into \eqref{eqn:obj} to obtain the concave lower bound $\underline p_{{\rm s}, 2}(\mathbf{q}_2)$, which is tight when $\mathbf{q}_2 = \mathbf{q}_2^o$. iii) Maximize the lower bound $ \underline p_{{\rm s},2}(\mathbf q_2)$ and finally obtain $ \mathbf{\underline{q}}_2^*$. } \label{fig:lowerbound}
\end{figure}

In Fig. \ref{fig:lowerbound}, we illustrate the procedure of finding the closed-form expression of the caching probability $\mathbf{\ubar q}_2^* \triangleq [\ubar q_{f,2}^*]_{f = 1, \cdots, N_f}$. Later, we show that the lower bound is tight, i.e., $\ubar p_{{\rm s},2}(\mathbf{q}_2)= p_{{\rm s}, 2}(\mathbf{q}_2)$, in low user density case, which means that $\mathbf{\ubar q}_{2}^*$ is the optimal solution of Problem 2 when user density is low. For general case, we will show that the caching probability $\ubar{q}_{f,2}^*$ can achieve almost the same success probability as the caching probability found by inter-point method in Section V. Since the computation of $\ubar{q}_{f,2}^*$ only requires twice bisection searches on two scalars, i.e., $\mu$ and $\nu$, it can be obtained with much lower complexity than the interior point method when $N_f$ is large.

Based on \eqref{eqn:optlower}, we can analyze the impact of user density on the optimal caching policy.

\begin{corollary}
	For any $\ubar{q}_{f,2}^*$, $\ubar{q}_{f+1,2}^* \in (0,1)$, $\ubar{q}_{f,2}^* - \ubar{q}_{f+1,2}^*$ increases with $\lambda_u/\lambda_2$.
\end{corollary}

\begin{IEEEproof}
See Appendix C.
\end{IEEEproof}

Corollary 1 indicates that when the ratio of user-to-helper density increases, the files with higher popularity have more chances to be cached and \emph{vice versa}, which reflects a trend towards caching the most popular files everywhere. On the contrary, when the ratio reduces, say from $\lambda_u/\lambda_2 \to \infty$ (implies no BS idling) to finite values (implies with BS idling), $\ubar{q}_{f,2}^* - \ubar{q}_{f+1,2}^*$ decreases. This indicates that BS idling makes the caching probability distribution less skewed, i.e, the diversity of cached files increases.

To further analyze how the per-user data rate requirement, BS density and transmit power of different tiers affect the caching policy, we derive the following corollaries in  extreme cases.
\subsubsection{High user density ($\lambda_u / \lambda_2$ is large so that $p_{{\rm a}, 2}(\mathbf{q}_2) \to 1$) }
In this case, all the helpers are active. The lower bound of success probability becomes
\begin{equation}
\ubar p_{{\rm s}, 2}(\mathbf q_{2})  = \sum_{f=1}^{N_f} \frac{ p_f q_{f,2}}{ \mathsf{C}_{1}(\bar\gamma_{0,2}) +  (\mathsf{C}_{3}(\bar\gamma_{0,2}) + 1)q_{f,2} }. \label{eqn:case1}
\end{equation}
Similar to  deriving \eqref{eqn:optlower}, we can obtain the  optimal caching probability maximizing $\ubar p_{{\rm s}}$ for 	$p_{{\rm a}, 2}(\mathbf{q}_2) \to 1$ as
\begin{equation}
	\ubar q_{f,2}^* =  \left[ \frac{\sqrt{\mathsf{C}_{1,\bar \gamma_{0,2}} + \mathsf{C}_{2}(\bar\gamma_{0,2})} }{\sqrt{\nu} (\mathsf{C}_{3}(\bar\gamma_{0,2}) + 1)} \sqrt{p_f}-\frac{\mathsf{C}_{1}(\bar\gamma_{0,2}) + \mathsf{C}_{2}(\bar\gamma_{0,2})}{\mathsf{C}_{3}(\bar\gamma_{0,2}) + 1} \right]_0^1, \label{eqn:opt high}
\end{equation}
where the Lagrange multiplier $\nu$ satisfying $ \sum_{f=1}^{N_f} \ubar q_{f,2}^* = N_c$ can be found by bisection searching.

\begin{corollary}
	When $R_0 \to \infty$ or $\lambda_u/\lambda_2 \to \infty$, $\ubar q_{1,2}^*, \cdots, \ubar q_{N_c,2}^* = 1$ and $\ubar q_{N_c+1,2}^*, \cdots, \ubar q_{N_f, 2}^* = 0$.
\end{corollary}

\begin{IEEEproof}
See Appendix D.
\end{IEEEproof}

Corollary 2 indicates that when the data requirement $R_0$ or $\lambda_u/\lambda_2$ is high, the optimal caching policy is simply caching the most popular files everywhere.

\subsubsection{Low user density ($\lambda_u/\lambda_2$ is low so that $U_{2} \to 1$ and  $p_{{\rm a}, 2}(\mathbf{q}_2) \to 0$)}
In this case, each helper serves at most one user and most of helpers have no user to serve, i.e.,  $p_{{\rm a}, 2}(\mathbf{q}_2) \to 0$, the success probability can be simplified into
\begin{equation}
p_{{\rm s}, 2}(\mathbf q_{2})  = \sum_{f=1}^{N_f} \frac{ p_f q_{f,2}}{ \mathsf{C}_{1}(\gamma_{0,2})  +  q_{f,2} } \label{eqn:case2},   
\end{equation}
where $\gamma_{0,2} = 2^{R_0/W} - 1$ because $U_2 \to 1$. Then, we have $\gamma_{0,2}(\mathbf{q}_2) = \bar \gamma_{0,2}$ and $p_{{\rm a},2}(\mathbf{q}_2) = \bar p_{{\rm a},2}$ and hence $\ubar p_{{\rm s}, 2}(\mathbf{q}_2) = p_{{\rm s}, 2}(\mathbf{q}_2)$, i.e., the lower bound is tight. Note that since $\gamma_{0,2}$ does not depend on $\mathbf q_2$ anymore in this case, $p_{{\rm s}, 2}(\mathbf q_2)$ is a concave function and hence Problem 2 becomes concave.

Again, similar to   deriving \eqref{eqn:optlower}, we can obtain the optimal caching probability maximizing $p_{{\rm s}, 2}(\mathbf{q}_2)$ for $U\to 1$ and $p_{{\rm a}, 2}(\mathbf{q}_2) \to 0$ as
\begin{equation}
q_{f,2}^* =  \left[ \frac{\sqrt{ \mathsf{C}_{1}(\gamma_{0,2}) }}{ \sqrt{\nu}}\sqrt{p_f}-\mathsf{C}_{1}(\gamma_{0,2})\right]_0^1, \label{eqn:opt}
\end{equation}
where $\nu$ satisfying $ \sum_{f=1}^{N_f}q_{f,2}^* = N_c$ can be found by bisection searching.
 We can prove that the conclusions in Corollary 2 also hold in this case, which are not shown for conciseness.

 In the following corollary, we show the impact of BS density and transmit power of different tiers on the  caching policy.

\begin{corollary}
	For any $q_{f,2}^*$, $q_{f+1,2}^* \in (0,1)$, $q_{f,2}^* - q_{f+1,2}^*$ increases with $\lambda_{12}$ and $P_{12}$.
\end{corollary}

\begin{IEEEproof}
See Appendix E.
\end{IEEEproof}

Corollary 3 indicates that when the MBS-to-helper density ratio $\lambda_1/\lambda_2$ or MBS-to-helper transmit power ratio $P_1/P_2$ increases, the files with higher popularity have more chances to be cached while the files with lower popularity have less chances to be cached, leading to a trend towards caching the most popular files everywhere. By contrast, when $\lambda_1/\lambda_2$ or $P_1/P_2$ decreases, the caching probability should be less skewed, i.e., the diversity of cached files increases.

We summary the impact of the system settings on the optimal caching policy in Table \ref{tab:impact}.

\begin{table*}[ht]
	\centering
	\caption{impact of the system setting on the optimal caching policy} 	\label{tab:impact}%
	\begin{tabular}{ccc}
		\toprule
		MBS-to-helper Transmit Power Ratio & MBS-to-helper Density Ratio & Optimal Caching Probability \\\hline
		$P_1/P_2$  $\downarrow$ & $\lambda_1/\lambda_2$ $\downarrow$     & Less Skewed \\
	$P_1/P_2$ 	$\uparrow$ &  $\lambda_1/\lambda_2$ $\uparrow$   & More Skewed \\ \hline
		User-to-helper Density Ratio, & Rate Requirement & Optimal Caching Probability \\ \hline
		$\lambda_u/\lambda_2$$\downarrow$ & $R_0$ $\downarrow$    & Less Skewed \\
		$\lambda_u/\lambda_2$ $\uparrow$ &  $R_0$ $\uparrow$   & More Skewed \\
		\bottomrule 
	\end{tabular}%
\end{table*}%

\section{Caching Policy Maximizing Area Spectral Efficiency}

In this section, we first derive the ASE of the cache-enabled HetNet as a function of caching probability and find the optimal caching probability maximizing the ASE. Then, we explain the difference of the optimal caching probability maximizing the success probability and the ASE. Finally, to show the merit of cache-enabled HetNet, we derived the ASE of traditional HetNets with PBSs with limited-capacity backhaul for numerical comparison.\footnote{Since when the rate requirement is high, e.g., larger than the backhaul capacity of each PBS, the success probability is zero when the user associates with the PBS tier, we did not compare the success probability of traditional HetNets equipped limited-capacity backhaul with cache-enabled HetNets for conciseness.} The ASE is computed in units of nats/s/Hz with $1$ nat/s = $1.443$ bps to simplify the expressions and analysis.

\subsection{Cache-enabled HetNets}
Based on the success probability derived in Proposition 1, we can obtain the ASE of cache-enabled HetNet in the following proposition.

\begin{proposition}
	The ASE of the cache-enabled HetNet with different association bias factor in the two tiers is
	\begin{align}
	&{\sf ASE}(\mathbf{q}_2)  =  \sum_{k=1}^{2} \lambda_k p_{{\rm a},k}(\mathbf{q}_2) M_k \sum_{f=1}^{N_f} \frac{p_fq_{f,k}}{\mathcal{P}_{k}(\mathbf{q}_2)} \nonumber \\ & \times \int_{0}^{\infty}  \left(\sum_{j=1}^{2}  \lambda_{jk} P_{jk}^{\frac{2}{\alpha}} \left(q_{f,j}  p_{{\rm a},j} (\mathbf{q}_2) B_{jk}^{\frac{2}{\alpha}} \mathsf{Z}_{1,jk}(e^x-1) \right.\right.\nonumber \\
	 & \left.\left.+( 1 -  q_{f,j})  p_{{\rm a},j}(\mathbf{q}_2) \mathsf{Z}_{2,jk}(e^x-1)  +  q_{f,j}   B_{jk}^{\frac{2}{\alpha}} \right) \vphantom{\sum_{j=1}^{2}} \right)^{-1}\!\!\!\!  {\rm d} x, \label{eqn:ASE}
	\end{align}
	where $ \mathsf{Z}_{1,kj} (x)$ and $\mathsf{Z}_{2,kj}(x)$ is defined in Proposition 1.
\end{proposition}

\begin{IEEEproof}
	See Appendix F.
\end{IEEEproof}

Then, considering the constraints on caching probability and cache size, the caching policy maximizing the ASE of cache-enabled HetNet can be found from
\begin{align}
\text{\em Problem 4:}\quad \max_{\mathbf q_2}~ & \mathsf{ASE}(\mathbf{q}_2)   \\
{\rm s.t.}~ & \eqref{eqn:con1},\eqref{eqn:con2} \nonumber
\end{align}

The computation of ASE requires $2N_f$ numerical integral coupled with $\mathbf{q}_k$, and hence  the optimal caching probability maximizing the ASE is hard to obtain even numerically. In the following, we first obtain a closed-form expression of an approximated ASE in a special case when the association bias factors are equal, and then try to find the optimal caching probability.

The key to obtain a closed-form expression of an approximated ASE is to derive a closed-form approximation for the integration in \eqref{eqn:ASE}, where $\mathsf{Z}_{1,jk}(e^x - 1)$ contains Gauss hypergeometric function. To tackle this problem, we first derive the asymptotic expressions of $\mathsf{Z}_{1,jk}(e^x - 1)$ and $\mathsf{Z}_{2,jk}(e^x - 1)$ for $x \to 0$ and $x\to \infty$ based on the asymptotic expressions of Gauss hypergeometric function and exponential function, respectively. Then, by substituting the asymptotic expressions into the integration,  we are able to derive a closed-form expression of an approximated ASE.

\begin{corollary}
When $B_1 = B_2$ and $N_c/N_f$ is large, the ASE of the cache-enabled HetNet can be approximated as
\begin{align}
\!\!{\sf ASE}(\mathbf{q}_2) \approx & \sum_{k=1}^{2}p_{{\rm a},k}(\mathbf{q}_2) \lambda_k  M_k \left(\ln 2  + \sum_{f=1}^{N_f}  \frac{p_f q_{f,k}}{\mathcal{P}_{k}(\mathbf{q}_2)} \right.\nonumber \\
&\left. \times \frac{\alpha}{2\mathsf{K}_{2,fk}(\mathbf{q}_2)} \ln \left( 1 + \frac{\mathsf{K}_{2,fk}(\mathbf{q}_2)}{\mathsf{K}_{1,k}(\mathbf{q}_2)}  4^{-\frac{1}{\alpha}}\right) \vphantom{\sum_{f=1}^{N_f}}\right), \label{eqn:ASEc}
\end{align}
where $\mathsf{K}_{1,k}(\mathbf{q}_2) \triangleq  \sum_{j=1}^{2}  \lambda_{jk} P_{jk}^{\frac{2}{\alpha}}  p_{{\rm a},j}(\mathbf{q}_2) \Gamma(1 - \frac{2}{\alpha}) \Gamma(M_j + \frac{2}{\alpha})\Gamma(M_j)^{-1} M_{jk}^{-\frac{2}{\alpha}}$, $\mathsf{K}_{2,fk}(\mathbf{q}_2) \triangleq \sum_{j=1}^{2}  \lambda_{jk} P_{jk}^{\frac{2}{\alpha}} (1 - p_{{\rm a},j}(\mathbf{q}_2)) $
\end{corollary}

\begin{IEEEproof}
See Appendix G.
\end{IEEEproof}
In Section V, we show that the approximation is also accurate even when $N_c/N_f$ is small (e.g., $N_c/N_f = 0.1$) by simulation.

Although the approximated ASE is in closed-form, the complex expressions of $\mathcal{P}_k(\mathbf{q}_2)$ and $p_{{\rm a},2}(\mathbf{q}_2)$ with respect to $\mathbf{q}_2$ make the optimization problem not concave in general. Only a local optimal solution can be found, e.g., by interior point method \cite{boyd2004convex}.

To obtain a closed-form expression for the caching probability, we consider an extreme case when $\lambda_u/\lambda_2 \to \infty$ and $\lambda_2/\lambda_1 \to \infty$. In this case, the active probability of helper is $p_{{\rm a},2} \to 1$ due to $\lambda_u/\lambda_2 \to \infty$ and the helper tier dominates the network ASE due to $\lambda_2/\lambda_1 \to \infty$. Then, we are able to simplify the expression of approximated ASE and obtain the following corollary.

\begin{corollary}
	When $\lambda_u/\lambda_2 \to \infty$ and $\lambda_2/\lambda_1 \to \infty$, the optimal caching probability maximizing the approximated ASE is $ q_{1,2}^*, \cdots,  q_{N_c,2}^* = 1$ and $ q_{N_c+1,2}^*, \cdots,  q_{N_f, 2}^* = 0$.
\end{corollary}

\begin{proof}
	See Appendix H.
\end{proof}

 Corollary 5 suggests that when the ratios of user-to-helper density and helper-to-MBS density are high, the optimal caching policy maximizing the ASE is to cache the most popular files everywhere, which is quite different from the caching policy maximizing the success probability. The reasons are explained as follows. Although increasing file diversity among helpers can increase the probability of user associated with the helper tier, it may make the user not able to associate with the closest helper, which decreases the signal power and increases the interference power and hence reduces the throughput of each helper. As a consequence, when the user-to-helper density and helper-to-MBS density are high so that all the helper node are active (i.e., each helper has at least one cache-hit user to serve) and the throughput of the helper tier dominates the overall ASE, increasing file diversity leads to low throughput of helper tier and hence low ASE. Therefore, in this case, caching the most popular files everywhere is the ASE-maximal caching policy.
 Intuitively, when the user-to-helper density is low such that some helper nodes may have no user to serve, increasing the file diversity among helpers can increase the probability of user associated with the helper tier and hence may increase the active probability of helper. Then, the throughput of each active helper may reduce due to the stronger interference. Nevertheless, we will show in Section V that even when the user-to-helper density is low, caching the most popular files everywhere can still achieve near-optimal ASE, because the overall ASE increases with more active helpers.

 \begin{remark}
For maximizing the success probability, caching the most popular files everywhere is optimal when user-to-helper density approaches infinity or rate requirement approaches infinity, and tends to become optimal when helper-to-MBS density or transmit power decreases. For maximizing ASE, caching the most popular files everywhere is optimal when both user-to-helper density and helper-to-user density approach infinity.
 \end{remark}

\subsection{Traditional HetNets}

For comparison, we derive the ASE of the traditional HetNet without local caching where a tier of MBSs with high-capacity backhaul is overlaid with a tier of PBSs with limited-capacity backhaul. For notational simplicity, we continue to use $\lambda_2$ and $P_2$ as the density and transmit power of PBS in this subsection. In traditional HetNets, we consider user association based on maximal BRP. Since each MBS and each PBS have backhaul and can retrieve all the files from the content server, traditional HetNet can be regarded as a special case of cache-enabled HetNet for $\mathbf{q}_2 = \mathbf{1}$ when compute the user association probability in Lemma 1 but with rate constraint when user associates the PBS tier as shown in the following.

Denote $C_{\rm bh, 2}$ as the backhaul capacity of each PBS and assume backhaul capacity is equally allocated among users. When the typical user associates with the PBS tier, its data rate is limited by the allocated backhaul capacity $\frac{C_{\rm bh, 2}}{U_2}$. Then, the average throughput of an active PBS can be expressed as
\begin{align}
\mathbb{E}\{R_2\} & =\mathbb{E}\left[ \sum_{u=1}^{U_2} \min\left\{\frac{W}{U_2/M_2}\ln(1 + \gamma_{u2}), \frac{C_{\rm bh, 2}}{U_2} \right\}\right]  \nonumber \\
& = \mathbb{E}\left[ \min\{W\ln(1 + \gamma_{u2}), C_{\rm bh, 2} \}\right], \label{eqn:spER2} 
\end{align}
where  function $\min\{x,y\} = x$ if $x \leq y$ and $\min\{x,y\} = y$ otherwise, and the last step follows because $\gamma_{u2}, u = 1, \cdots, U_2$ are independently and identically distributed and substituting $M_2 = 1$.

Similar to deriving  \eqref{eqn:ASE} and considering the limited-capacity backhaul of PBS tier, we can obtain the ASE of traditional HetNet as follows.

\begin{proposition}
	The ASE of traditional HetNet is
	\begin{align}
	{\sf ASE}  = & \sum_{k=1}^{2} p_{{\rm a}, k} \lambda_k M_k\frac{1}{\mathcal{P}_k}  \int_{0}^{\frac{C_{{\rm bh}, k}}{W}} \left( \sum_{j=1}^{2}  \lambda_{jk} (P_{jk}B_{jk})^{\frac{2}{\alpha}} \right. \nonumber\\
	& \left. \times\left( p_{{\rm a},j} \mathsf{Z}_{1,jk}(e^x-1)   +  1 \right) \vphantom{\sum_{j=1}^{2}} \right)^{-1} {\rm d} x,
	\end{align}
	where the backhaul capacity of each MBS $C_{{\rm bh}, 1} \to \infty$, $\mathcal{P}_k = \mathcal{P}_k(\mathbf{q}_2)|_{\mathbf{q}_2 = \mathbf{1}}$, $p_{{\rm a},j} = p_{{\rm a},j}(\mathbf{q}_2)|_{\mathbf{q}_2 = \mathbf{1}}$ and $\mathsf{Z}_{1,jk}(x)$ is defined in Proposition 3.
\end{proposition}

\begin{IEEEproof}
	See Appendix I.
\end{IEEEproof}
Comparing Proposition 3 and Proposition 4, we can see that when the backhaul capacity of PBS is unlimited in traditional HetNet and each helper caches all the files in cache-enabled HetNet, the ASE of cache-enabled HetNet and traditional HetNet are equivalent.

The computation of the ASE of traditional HetNets requires double numerical integration. To reduce the computational complexity, similar to deriving Corollary 4, we obtain the closed-form expression of an approximated ASE of traditional HetNet in a special case.

\begin{corollary}
When $B_1 = B_2$ and $\frac{C_{\rm bh, 2}}{W} \ll \ln 2$, the ASE of traditional HetNet can be approximated as
\begin{align}
{\sf ASE}  \approx & \lambda_1 M_1  \left(\ln 2 +\frac{\alpha}{2\mathcal{P}_1 \mathsf{K}_{2}}  \ln \left(1 + \frac{\mathsf{K}_2}{\mathsf{K}_1}4^{-\frac{1}{\alpha}}\right)   \right)  \nonumber \\
&+\lambda_2 p_{{\rm a},2} C_{\rm bh, 2}, 
\end{align}
where $\mathsf{K}_{1} \triangleq  \sum_{j=1}^{2}  \lambda_{j1} P_{j1}^{\frac{2}{\alpha}}  p_{{\rm a},j} \Gamma(1 - \frac{2}{\alpha}) \Gamma(M_j + \frac{2}{\alpha})\Gamma(M_j)^{-1} $ $M_{j1}^{-\frac{2}{\alpha}}$, and $\mathsf{K}_{2} \triangleq \sum_{j=1}^{2}  \lambda_{j1} P_{j1}^{\frac{2}{\alpha}} (1 - p_{{\rm a},j})$.
\end{corollary}

\begin{IEEEproof}
See Appendix J.
\end{IEEEproof}
In Section V, we show that the approximation is accurate as long as $\frac{C_{\rm bh, 2}}{W} \leq \ln 2$ (i.e., $C_{\rm bh, 2} \leq 20$ Mbps for $20$ MHz transmission bandwidth). Considering that copper lines are the most widely used backhaul solutions \cite{dahrouj2015cost} with average data rate of 5 Mbps in the US \cite{Andy2012} and its cost grows linearly with the capacity \cite{dahrouj2015cost}, Corollary 6 can cover most densely deployed small BSs within affordable cost.


\section{Simulation and Numerical Results}

In this section, we validate previous analysis via simulation, and demonstrate how different factors affect the optimal caching policies, as well as the corresponding success probability and ASE of cache-enabled HetNet via numerical results.

The following caching policies are considered for comparison:
\begin{enumerate}
	\item ``{\em Max $p_{\rm s}$ (Inter-point)}": The solution of Problem 1 found by interior point method. Since the problem is not concave, we try 100 different random initial values and pick the best solution to increase the opportunity to find the global optimal solution.
	\item ``{\em Max $p_{{\rm s}, 2}$ (Lower-bound)}": The optimal solution that maximizes the lower bound $\ubar p_{{\rm s}, 2}(\mathbf{q}_2)$ in Proposition 2, i.e., $\ubar q_{f,2}^*$.
	\item  ``{\em Max ASE}": The solution of Problem 4 found by interior point method with 100 different random initial values.
	\item  ``{\em Popular}": Caching the most popular files everywhere, i.e., $q_{1,2}, \cdots, q_{N_c, 2} = 1$ and $q_{N_c + 1,2}, \cdots,$ $ q_{N_f, 2} = 1$. This policy achieves no file diversity.
	\item  ``{\em Uniform}":  Each file is cached with equal caching probability, i.e., $q_{f,2} = N_c/N_f$. This policy achieves the maximal file diversity.
\end{enumerate}

Unless otherwise specified, the following setting is used. We consider a $3$ km $\times$ $3$ km square area. The MBS, helper, and user densities are $\lambda_1 = 1/(250^2 \pi)$ m$^{-2}$, $\lambda_2 = 50/(250^2 \pi)$ m$^{-2}$ and $\lambda_u = 50/(250^2 \pi)$ m$^{-2}$, respectively. The path-loss exponent is $\alpha = 3.7$ \cite{3GPP}. The transmit power of each MBS (with $M_1 = 4$ antennas) and helper are $P_1 = 46$ dBm and $P_2 = 21$ dBm, respectively \cite{3GPP}. The bias factors for the MBS and helper tiers are $B_1 = 1$ and $B_2 = 10$, respectively. The transmission bandwidth is $W = 20$ MHz and the minimal rate requirement for each user is $R_0 = 2$ Mbps. The file catalog size is $N_f = 1000$ files and each helper can cache $N_c = 100$ files, i.e., the normalized cache size is $\eta \triangleq N_c/N_f = 10\%$.

\subsection{Success Probability}
\begin{figure}[!htb]
	\centering
	\includegraphics[width=\textwidth]{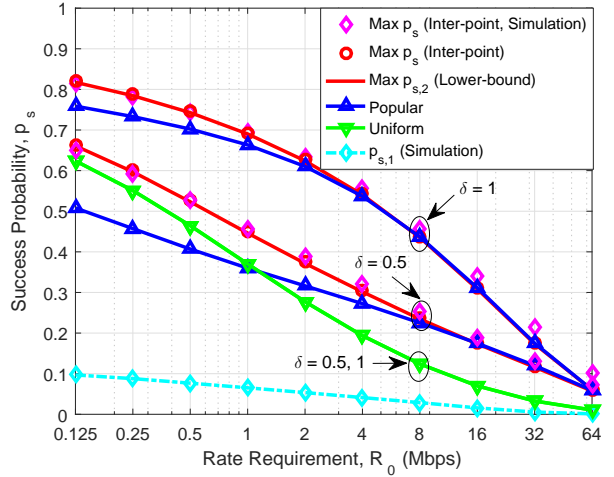}
	\caption{Validation of the analysis with simulation.} \label{fig:SINR}
\end{figure}

In Fig. \ref{fig:SINR}, we show the simulation and numerical results of success probability versus the per-user rate requirement. The simulation result is obtained based on $q_{f,2}^*$ from the interior point method and then by computing $p_{\rm s} (\mathbf q_2^*)$ via Monte Carlo method considering $-95$ dBm noise power. The numerical results are computed from Proposition 1. We can see that the numerical results (with legend ``Max $p_{\rm s}$ (Inter-point)'') is slightly lower than the simulation results (with legend ``Max $p_s$ (Inter-point, Simulation)'') when $R_0$ is low and slightly higher when $R_0$ is high, because Proposition 1 is derived in interference-limited scenario and by using the mean load approximation. For all the considered caching policies, the simulation results are close to the numerical results, which are not shown for a clean figure. Hence, in the rest of this subsection we only provide the numerical results. We also provide the simulation result of the success probability contributed by the MBS tier, i.e., $p_{{\rm s},1}(\mathbf{q}_2^*)$, which is less than 0.05 when $R_0 > 1$ Mbps and hence we can safely neglect the impact of  the success probability contributed by the MBS tier. Comparing {``Max $p_{{\rm s},2}$ (Lower-bound)''} with {``Max. $p_{\rm s}$ (Inter-point)''}, we can observe that the caching policy maximizing the lower bound of success probability contributed by the helper tier $\ubar p_{{\rm s}, 2}(\mathbf{q}_2)$ performs almost the same as the caching policy maximizing the total success probability $p_{\rm s}(\mathbf{q}_2)$. When the rate requirement is high, e.g., $R_0 > 4$ Mbps for $\delta = 0.5$ or $R_0 > 2$ Mbps for $\delta = 1$, caching the most popular files everywhere can achieve almost the same performance as the optimized caching policies in this setup with $\lambda_u/\lambda_2=1$, which validates Corollary 2 although it is derived when $R_0 \to \infty$ and $\lambda_u/\lambda_2 \to \infty$. Moreover, when $\delta$ increases, the gap between caching the most popular files everywhere and the optimized caching policies shrinks. These results indicate  that for highly skewed demand with high rate requirement, e.g., video on demand service, simply caching the most popular files everywhere can achieve maximal success probability.

\begin{figure}[htb]
	\centering
	\includegraphics[width=\textwidth]{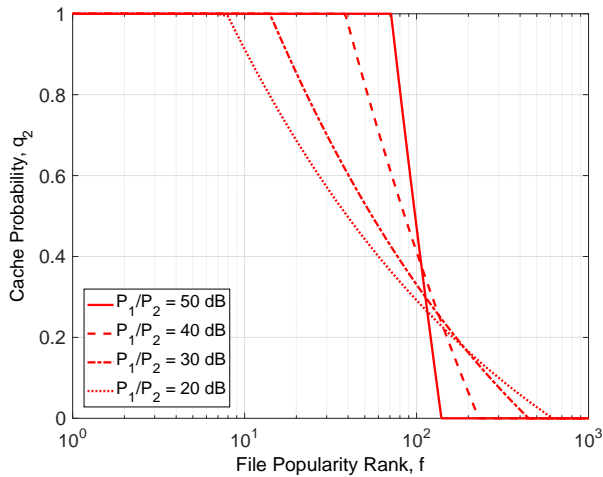}
	\caption{Impact of difference in transmit powers, $\delta = 0.5$.} 
	\label{fig:P}
\end{figure}

In Fig. \ref{fig:P}, we show the impact of difference in transmit powers on the optimal caching policy. When $P_{1}/P_2$ increases, the files with higher popularity have more chances to be cached while the files with lower popularity have less chances to be cached, which agrees with Corollary 3 although it is derived when $\lambda_u/\lambda_2 \to 0$. This is because the interference from the MBS increases with $P_1/P_2$ and caching the most popular file everywhere can increase the user's SINR and hence increase the rate, which leads to the increase of success probability.

\begin{figure}[!htb]
	\centering
	\includegraphics[width=\textwidth]{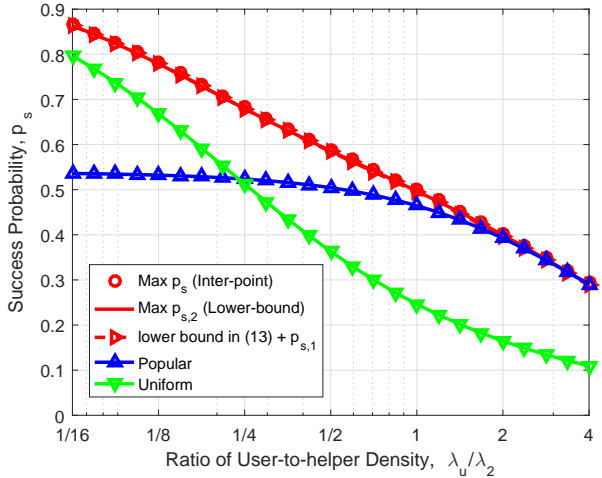}
	\caption{Impact of helper density, $\delta = 0.5$.} \label{fig:lambda}
\end{figure}

%
%

In Fig. \ref{fig:lambda}, we show the impact of helper density with given MBS and user densities. We can see that the gain of optimal caching over caching popular files everywhere shrinks when user-to-helper density $\lambda_u/\lambda_2$ increases ($\lambda_1/\lambda_2$ increases as well), which agrees with Corollary 1 and Corollary 3. When $\lambda_u/\lambda_2$ is low, the gain of optimal caching over uniform caching approaches to zero. This can be explained as follows. When $\lambda_2$ increases, the SINR at the user associated with the helper tier increases since the helper is closer to the user meanwhile more helpers can be turned into idle mode leading to lower interference. On the other hand, the number of users served by each helper will decrease, hence the time-frequency resources allocated to each user will increase. As a result, the achievable rate of each user will increase and hence the optimal caching probability is less skewed to increase file diversity. Moreover, we can see that the lower bound of success probability computed from \eqref{eqn:Pofflower}, i.e., $\ubar p_{{\rm s},2 }(\mathbf{q}_2^*) + p_{{\rm s}, 1}(\mathbf{q}_2^*)$, almost overlaps with ``Max $p_{s,2}$ (Inter-point)", i.e., $p_{\rm s}(\mathbf{q}_2^*) = p_{\rm s, 2}(\mathbf{q}_2^*) + p_{\rm s, 1}(\mathbf{q}_2^*)$, which means that the lower bound in \eqref{eqn:Pofflower} is also tight for more general user density case.

\begin{figure}[!htb]
	\centering
	\includegraphics[width=\textwidth]{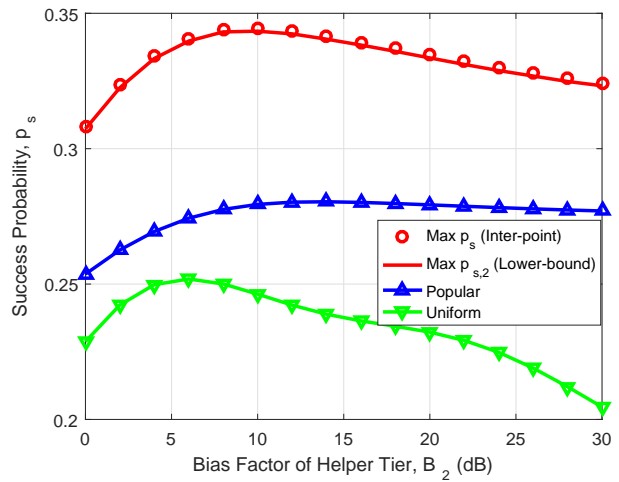}
	\caption{Impact of bias factor, $B_1 = 1$, $\delta = 0.5$.} \label{fig:B}
\end{figure}

In Fig. \ref{fig:B}, we show the impact of the bias factor. We can see that the optimized caching policies outperform the other two policies. When $B_2$ increases, the success probability first increases and then decreases (though slightly for ``Popular"). This is because the users are more likely to associate with the helper tier when $B_2$ increases, and hence increases the success probability. However, when $B_2$ continues to increase, the number of users served by each helper also increases, leading to less time-frequency resources allocated to each user. Meanwhile, the distance between the user and its interfering MBS increases since the user prefers a far helper node to associate with than a near MBS, which reduces the SINR of user. Therefore, the success probability finally decreases with $B_2$.

\subsection{Area Spectral Efficiency}
\begin{figure}[!htb]
	\centering
	\includegraphics[width=\textwidth]{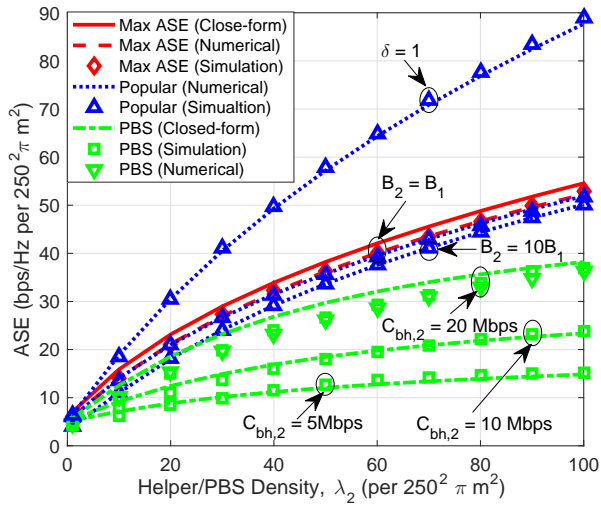}
	\caption{ASE v.s. helper/PBS density $\lambda_2$, $\delta = 0.5$, $B_2 = B_1$ (unless otherwise specified).} \label{fig:ASE_lambda}
\end{figure}

In Fig. \ref{fig:ASE_lambda}, we show the relation between ASE and helper density and compare the ASE of cache-enabled HetNet with traditional HetNet (with legend ``PBS"). The numerical results are obtained from Proposition 3 for cache-enabled HetNet and Proposition 4 for traditional HetNet, the results of the closed-form expression are obtained from Corollary 4 for cache-enabled HetNet and Corollary 6 for traditional HetNet, and the simulation results are obtained by Monte Carlo method considering $-95$ dBm noise power. We can see that the numerical results derived in interference-limited scenario are almost overlapped with the simulation results for non-biased user association, i.e., $B_2 = B_1$ or biased user association, e.g., $B_2 = 10 B_1$. The results obtained from closed-form expressions of approximated ASE are slightly higher than the simulation results when $B_2  = B_1$. The gap is mainly from the approximation in \eqref{eqn:int1} and \eqref{eqn:ER2b} where we neglect some small positive terms in the denominator. Nevertheless, the closed-form results are quite accurate and hence, we only provide the closed-form results in the sequel for conciseness. For traditional HetNets, the closed-form results are quite accurate when $C_{\rm bh, 2} = 20$ Mbps (i.e., $C_{\rm bh, 2}/W = \ln 2$) although Corollary 6 is derived by assuming $C_{\rm bh, 2}/W \ll \ln2 $.

Moreover, caching the most popular files everywhere can achieve almost the same ASE as the optimal caching policy maximizing the ASE and the gap is minor even when  $\lambda_u/\lambda_2$ is small, say $\lambda_u/\lambda_2 = 0.5$ when $\lambda_2 = 100/(250^2 \pi)$, which validates Corollary 5 although it is derived when $\lambda_u/\lambda_2 \to \infty$ and $\lambda_2/\lambda_1 \to \infty$. Besides, non-bias user association can achieve higher ASE compared with $B_2 = 10 B_1$. This is because for biased user association, some users may associate with the helper that offers weaker received signal than a MBS, which reduces the SINR of the users and hence may decrease the ASE.

Compared with the traditional HetNet with limited-capacity backhaul (e.g.,  $C_{\rm bh, 2}=10$ or 20 Mbps), the cache-enabled HetNet can double the ASE when each helper node only
caches 10\% of the total files (e.g., when $\delta = 0.5$ or $\delta = 1$). Alternatively, to achieve the same ASE, the
helper node density is much lower than the PBS density, which can
reduce the deployment and operation cost remarkably (e.g., when
$C_{\rm bh, 2} = 10$ Mbps, the helper node density is
about 1/3 of the PBS density to achieve an ASE of $20/(250^2\pi)$ bps/Hz/m$^2$).

\begin{figure}[!htb]
	\centering
	\includegraphics[width=\textwidth]{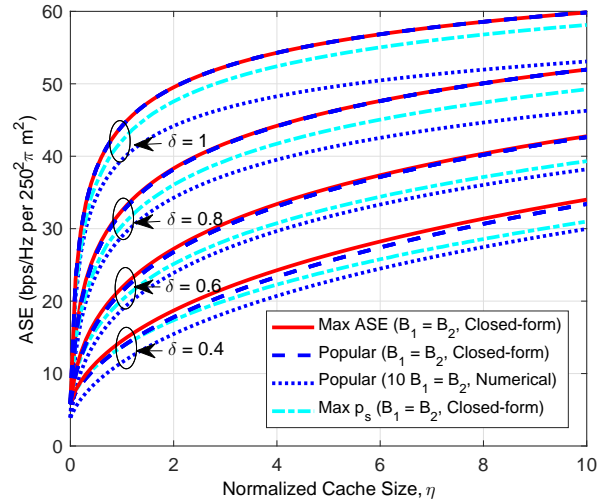}
	\caption{ASE v.s. normalized cache size $\eta$.} \label{fig:ASE_eta}
\end{figure}
In Fig. \ref{fig:ASE_eta}, we show the relation between ASE and normalized cache size with different skew parameter $\delta$. As expected, the ASE of cache-enabled HetNet increases with $\eta$. Furthermore, when $\delta$ is larger, the ASE grows with $\eta$  more rapidly for small value of $\eta$ and grows with $\eta$ more slowly for large $\eta$. Moreover, caching the most popular files everywhere can still achieve almost the maximal ASE when $\eta$ and $\delta$ change, especially when $\delta$ is large. Besides, by comparing with the results of Fig. \ref{fig:ASE_lambda}, we can see that non-biased user association biases can achieve higher ASE than $B_2 = 10 B_1$ for different helper density and cache size. Therefore, in the following, we only consider caching the most popular files everywhere with non-biased user association. We also show the ASE based on the caching probability maximizing the success probability when $R_0 = 0.25$ Mbps, which is slight lower than the maximal ASE. By comparing with Fig. \ref{fig:SINR}, we can conclude that when $R_0$ and $\delta$ are high, caching policy maximizing ASE is equivalent to caching maximizing the success probability, while when $R_0$ and $\delta$ are low, the caching policy maximizing ASE suffers a more severe degradation in success probability than the ASE degradation of caching policy maximizing success probability.

\begin{figure}[!htb]
	\centering
	\includegraphics[width=\textwidth]{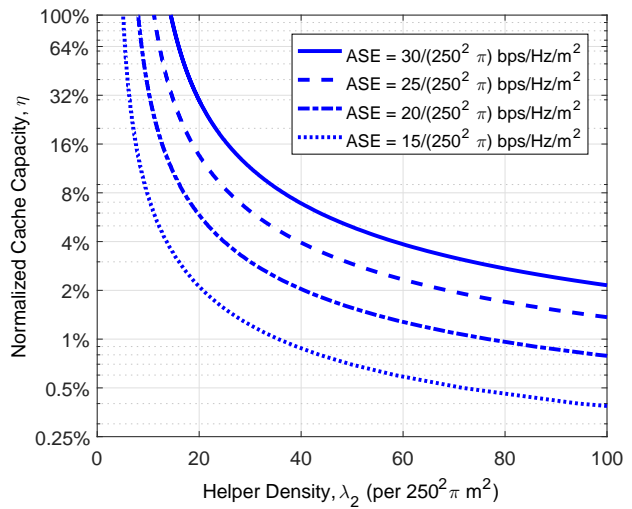}
	\caption{Trade-off between helper density and cache size, $\delta = 0.5$, $B_2 = B_1$.} \label{fig:tradeoff}
\end{figure}

Since the ASE of cache-enabled HetNet can be improved either by increasing
cache size or helper density, a natural question is: how helper
density can be traded off by cache size to achieve a target
ASE? To answer this question, we set the ASE as different values and show
the normalized cache size versus helper density in Fig.
\ref{fig:tradeoff}. With a given target ASE and cache size $N_c$, helper density $\lambda_2$ can be found from \eqref{eqn:ASEc} using the bisection search method. We can observe for example that
to achieve a target ASE of $20/(250^2 \pi)$ bps/Hz/m$^2$, by increasing
the cache size from $\eta = 1\%$ to $\eta = 2\%$, the helper density can be reduced by half. Similar trade-off between BS density and cache size was reported in \cite{EURASIP}, where a homogeneous network with single antenna BSs was considered and the performance metric was success probability. When use ASE as the metric, our results show that the helper density can be reduced more significantly by increasing the cache size.

\begin{figure}[!htb]
	\centering
	\includegraphics[width=\textwidth]{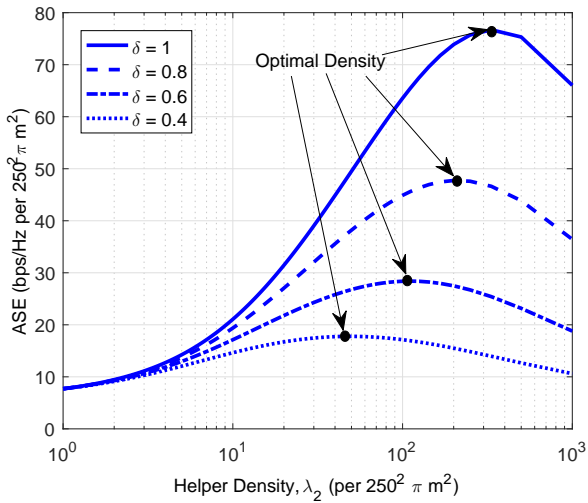}
	\caption{ASE v.s. helper  density with given $\lambda_2N_c  = 10^3/(250^2\pi)$, $B_2 = B_1$.} \label{fig:delta}
\end{figure}

Inspired by such a trade-off, another natural question is: with a given total amount of cache size
within an area, should we deploy the caches in a distributed manner (i.e., more helpers each with small cache size) or in a centralized manner (i.e.,
less helpers each with large cache size) in order to maximize
the ASE? To answer this question, we fix the area cache size
as a constant, e.g., $\lambda_2N_c  = 10^3/(250^2\pi)$, and provide the ASE versus the helper
density in Fig. \ref{fig:delta}. When $\lambda_2 = 1/(250^2\pi)$, every helper caches all the files, and when $\lambda_2 = 10^3/(250^2\pi)$, each helper caches only one file. We can see that there exists an
optimal helper density maximizing the ASE. Moreover, the optimal  density increases with $\delta$, which means that the more skewed the file
popularity is, the more distributedly we should deploy the caches.
This can be explained from the impact of the following two observations in Figs. \ref{fig:ASE_lambda} and \ref{fig:ASE_eta}. On one
hand, with given cache size of each helper, the ASE increases with
the helper density first rapidly and then slowly. On the other hand, with
given helper density, the ASE reduces with the decrease of cache
capacity first slowly and then rapidly, and the larger $\delta$ is, the more slowly the
ASE decreases with $\eta$ when the cache size is large.

\section{Conclusion}

In this paper, we investigated the maximal success probability and ASE of cache-enabled HetNets. Under the probability caching framework, we obtained the optimal caching probability  respectively maximizing the success probability and ASE, and analyzed the impact of system settings. Simulation results validated our analysis. The following conclusions can be drawn from the analytical and numerical results. To maximize the success probability, when the MBS-to-helper
transmit power ratio,  MBS-to-helper density ratio, user-to-helper density ratio or the per-user rate requirement is low, the optimal caching probability is less skewed that achieves file diversity among helpers. Otherwise, caching probability is more skewed. To maximize the ASE, each helper tends to cache the most popular files. Besides, there exists an optimal bias factor to maximize the success probability. By turning the helpers with no user to serve into idle mode, the caching probability should be less skewed to increase file diversity. Compared with traditional HetNet, the helper density is much lower than the PBS density to achieve the same target ASE, and the helper density can be further reduced by increasing cache size. With given total cache size within an area, there exists an optimal helper density that maximizes the ASE.
\appendices

\section{Proof of Lemma 1}

\setcounter{equation}{0}
\renewcommand{\theequation}{A.\arabic{equation}}
Using the law of total probability, we can obtain $
\mathcal{P}_{k}(\mathbf{q}_2) =  \sum_{f = 1}^{N_f} p_k \mathcal{P}_{f,k}(\mathbf{q}_2)$,
where $\mathcal{P}_{f,k}(\mathbf{q}_2)$ is the probability of the typical user associated with the $k$th tier conditioned on that it requests the $f$th file. Since the user requesting the $f$th file associates with the BS with the strongest BRP in $\{\Phi_{f,k}\}$ $(k = 1,2)$, which follows independent homogeneous PPP with density $\{q_{f,k} \lambda_k\}$ $(k = 1,2)$ as we mentioned before, $\mathcal{P}_{f,k}(\mathbf{q}_2)$ can be obtained from \cite[Lemma 1]{flexible} as
$
\mathcal{P}_{f,k}(\mathbf{q}_2)  = q_{f,k}\left(\sum_{j=1}^{2}q_{f,j}\lambda_{jk} (P_{jk}B_{jk})^{\frac{2}{\alpha}} \right)^{-1}. \label{eqn:con_asso 2}
$
Then, Lemma 1 can be proved. \qed

\section{Proof of Proposition 1}
\renewcommand{\theequation}{B.\arabic{equation}}
\setcounter{equation}{0}
From \eqref{eqn:eql}, by defining $ \gamma_{0,f,k} \triangleq 2^{\frac{R_0 U_{f,k}}{WM_{k}}} - 1$ as the equivalent SINR requirement when the user requests the $f$th file and associates with the $k$th tier, we obtain
\begin{align}
p_{\rm s}(\mathbf{q}_2) & =  \sum_{k = 1}^{2}  p_{{\rm s}, k}(\mathbf{q}_2) \nonumber \\
& =  \sum_{k = 1}^{2} \sum_{f=1}^{N_f} p_f \mathcal{P}_{f,k}(\mathbf{q}_2) \mathbb{P} \left( \gamma_{f,k}(\mathbf{q}_2) \geq \gamma_{0,f,k} \right),  \label{eqn:off}
\end{align}

Then, based on the law of total probability, we have
\begin{align}
&\mathbb P (\gamma_{f,k}(\mathbf{q}_2) > \gamma_{0,f,k} ) \nonumber \\
 & = \int_{0}^{\infty}\mathbb{P} (\gamma_{f,k}(\mathbf{q}_2) > \gamma_{0,f,k} ~|~ r) f_{r_k,\mathbf{q}_2}(r) {\rm d} r,  \label{eqn:suc3}
\end{align}
where $f_{r_{k},\mathbf{q}_2}(r) = \frac{2 \pi q_{f,k}\lambda_k}{\mathcal{P}_{f,k}} r  \exp \big(-\pi r^2 \sum_{j=1}^{2} q_{f,j}\lambda_j $ $ ({P}_{jk}B_{jk})^{\frac{2}{\alpha}} \big)$ is the probability density function of the distance between user and its serving BS when requesting the $f$th file and associated with the $k$th tier \cite{flexible}, and $\mathbb P (\gamma_{f,k}(\mathbf{q}_2) > \gamma_{0,f,k} ~|~ r )$ is the success probability conditioned on that the distance between the typical user and its serving BS is $r$, which can be derived as,
\begin{align}
&\mathbb{P} (\gamma_{f,k}(\mathbf{q}_2) > \gamma_{0,f,k} ~|~ r) \nonumber \\
\overset{(a)}{=}& \mathbb{E}_{U_{f,k},I_k} \left[  \mathbb{P}\left( h_{k0} > M_k P_k^{-1} r^{\alpha} I_k \gamma_{0,f,k} ~\big|~r, I_k, U_{f,k}\right) \right] \nonumber \\
 \overset{(b)}{=}& \mathbb{E}_{U_{f,k}, I_k} \left[\exp(-M_kP_k^{-1}r^\alpha I_k \gamma_{0,f,k})\right] \nonumber\\
 \overset{(c)}{=}&\mathbb{E}_{U_{f,k}} \left[ \prod_{j=1}^{2}  \mathcal{L}_{I_{f,kj}} \left(M_kP_k^{-1}r^\alpha \gamma_{0,f,k} \right) \right.\nonumber\\
&\left. \times  \prod_{j=1}^{2} \mathcal{L}_{I_{f',kj}} \left(M_kP_k^{-1}r^\alpha \gamma_{0,f,k} \right)\right] \nonumber \\
\overset{(d)}{=} &\sum_{U_{f,k} = n}^{\infty} \mathbb{P} (U_{f,k} = n)\left( \prod_{j=1}^{2}  \mathcal{L}_{I_{f,kj}} \left(M_kP_k^{-1}r^\alpha \gamma_{0,f,k} \right) \right.\nonumber \\
& \left.\times \prod_{j=1}^{2} \mathcal{L}_{I_{f',kj}} \left(M_kP_k^{-1}r^\alpha \gamma_{0,f,k} \right)\right) \Bigg|_{U_{f,k} = n},\label{eqn:con_suc}
\end{align}
where $U_{f,k}$ is the number of users served by the same BS together with the typical user when the typical user requests the $f$th file and associates with the $k$th tier, step $(a)$ is from \eqref{eqn:gamma} by neglecting the thermal noise, step $(b)$ is from $h_{k0} \sim \exp(1)$, step $(c)$ follows because $\mathcal{L}_{\sum_{j}I_{f,kj}}(s)
=\prod_{i} \mathcal{L}_{I_{f, kj}}(s)$, and $\mathcal{L}(\cdot)$ denotes the Laplace
transform, step $(d)$ is from the law of total probability.

Upon employing the  BS active probability approximation, the distribution of the active BSs in the $j$th tier caching the $f$th file  $\tilde{\Phi}_{f,k}(\mathbf{q}_2)$ and those not caching the $f$th file  $\tilde{\Phi}_{f',j}(\mathbf{q}_2)$ are two homogeneous PPP with density $p_{{\rm a},j}q_{f,j}\lambda_j$ and $p_{{\rm a},j}(1-q_{f,j})q_{f,j}\lambda_j$. Then, the Laplace transform of $I_{f,kj}$ in \eqref{eqn:con_suc} can be derived as
\begin{align}
&   \mathcal{L}_{I_{f,kj}} (s) = \mathbb{E}_{\tilde \Phi_{f,j}(\mathbf{q}_2), h_{ji}} \left[e^{-s \sum_{i\in \tilde\Phi_{f,j}(\mathbf{q}_2) \backslash b_{k0}} P_j  h_{ji} r_{ji}^{-\alpha_j}}\right] \nonumber\\
& \overset{(a)}{=}  \mathbb{E}_{\tilde \Phi_{f,j}(\mathbf{q}_2)}\left[ \prod_{i\in \tilde\Phi_{f,j}(\mathbf{q}_2) \backslash b_{k0}}\left(1 + \frac{s P_j}{M_j}  r_{j,i}^{-\alpha_j}  \right)^{-M_j}\right]  \nonumber \\
&   \overset{(b)}{=} e^{ -2\pi p_{{\rm a}, j} q_{f,j}\lambda_j \int_{r_{0j}}^{\infty} \left(1 - \left({1 + \frac{s P_j}{M_j} v^{-\alpha_j}} \right)^{-M_j}\right)v {\rm d}v} \nonumber \\
& = e^{ - \pi p_{{\rm a},j} q_{f,j} \lambda_j r_{0j}^2 \left(F_1 \left[ -\frac{2}{\alpha}, M_j; 1-\frac{2}{\alpha}; -\frac{sP_j}{M_j} r_{0j}^{-\alpha}\right] -1\right)  }, \label{eqn:Laplace 1}
\end{align}
where step $(a)$ follows from $h_{ji}\sim \mathbb G(M_j,1/M_j)$, step $(b)$ follows from the probability generating function of the PPP, and $r_{0j} = (P_{jk}B_{jk})^{\frac{1}{\alpha}} r$ is the closest possible distance of the interfering BS in $\tilde \Phi_{f,j}(\mathbf{q}_2)$.
Substituting $r_{0j}$ and $s = M_k P_k^{-1} r^\alpha \gamma_{0,f,k}$ into \eqref{eqn:Laplace 1}, we obtain
\begin{equation}
\mathcal{L}_{f,kj}(M_k P_k^{-1} r^\alpha \gamma_{0,f,k})
= e^{  - \pi p_{{\rm a},j} q_{f,j}\lambda_j  (P_{jk}B_{jk})^{\frac{2}{\alpha}} r^2 \mathsf{Z}_{1,kj} (\gamma_{0,f,k})}, \label{eqn:L1}
\end{equation}
where $ \mathsf{Z}_{1,kj} (x)\triangleq {}_{2}F_1 \big[ -\frac{2}{\alpha}, M_j; 1-\frac{2}{\alpha}; -\frac{x}{M_{jk}B_{jk}} \big] -1$.

Since the BSs not caching the $f$th file, i.e. $\tilde{\Phi}_{f',j}$, can be arbitrarily close to the user, i.e., $r_{0j} \to 0$, from $\lim\limits_{r_{0j} \to 0}  r_{0j}^2(_2F_1 \big[-\frac{2}{\alpha}, M_j; 1-\frac{2}{\alpha}; -\frac{sP_j}{M_j} r_{0j}^{-\alpha}\big] -1 ) = \Gamma (1-\frac{2}{\alpha}) \Gamma(M_j + \frac{2}{\alpha_j}) \Gamma(M_j)^{-1}(\frac{sP_j}{M_j})^{\frac{2}{\alpha}}$, similar to the derivation of \eqref{eqn:Laplace 1}, we can obtain
\begin{equation}
\mathcal{L}_{f',kj}(M_k P_k^{-1} r^\alpha \gamma_{0,f,k})= e^{-\pi p_{{\rm a},j} (1-q_{f,j})\lambda_j P_{jk}^{\frac{2}{\alpha}} r^2 \mathsf{Z}_{2,kj}(\gamma_{0,f,k}) }, \hspace{-1mm} \label{eqn:L2}
\end{equation}
where $\mathsf{Z}_{2,kj}(x) \triangleq \Gamma \left(1-\frac{2}{\alpha}\right)  \Gamma\left(M_j + \frac{2}{\alpha}\right) \Gamma(M_j)^{-1}(\frac{x}{M_{jk}})^{\frac{2}{\alpha}}$.

By substituting \eqref{eqn:L1} and \eqref{eqn:L2} into \eqref{eqn:con_suc}, further considering \eqref{eqn:suc3} and $\int_0^\infty 2r e^{-Ar^2}{\rm d} r = \frac{1}{A}$, we obtain
\begin{align}
&\mathbb{P} ( \gamma_{f,k}(\mathbf{q}_2) > \gamma_{0,f,k} )=  \sum_{n=1}^{\infty} \mathbb{P}(U_{f,k} = n)\frac{q_{f,k}}{\mathcal{P}_{f,k}(\mathbf{q}_2)}\nonumber\\
& \times \left(\sum_{j=1}^{2}  \lambda_{jk} P_{jk}^{\frac{2}{\alpha}} \left( q_{f,j}  p_{{\rm a},j}(\mathbf{q}_2)  B_{jk}^{\frac{2}{\alpha}} \mathsf{Z}_{1,jk}(\gamma_{0,f,k}) +(1- q_{f,j}) \right. \right.  \nonumber\\
& \left. \times p_{{\rm a},j}(\mathbf{q}_2)\mathsf{Z}_{2,jk}(\gamma_{0,f,k})  +  q_{f,j}   B_{jk}^{\frac{2}{\alpha}} \right) \Bigg)^{-1} \Bigg|_{U_{f,k} = n}. \label{eqn:suc}
\end{align}
Upon using the average load approximation, \eqref{eqn:suc} becomes
\begin{align}
&\mathbb{P} ( \gamma_{f,k}(\mathbf{q}_2) > \gamma_{0,f,k} )\nonumber\\
& =   \frac{q_{f,k}}{\mathcal{P}_{f,k}(\mathbf{q}_2)} \left(\sum_{j=1}^{2}  \lambda_{jk} P_{jk}^{\frac{2}{\alpha}} (q_{f,j}  p_{{\rm a},j}(\mathbf{q}_2) B_{jk}^{\frac{2}{\alpha}} \mathsf{Z}_{1,jk}(\gamma_{0,2}(\mathbf{q}_2)) \right. \nonumber\\
&+ (1-q_{f,j})  p_{{\rm a},j}(\mathbf{q}_2)\mathsf{Z}_{2,jk}(\gamma_{0,2}(\mathbf{q}_2))  +  q_{f,j}   B_{jk}^{\frac{2}{\alpha}} )\Bigg)^{-1}, \label{eqn:suc2}
\end{align}
where $\gamma_{0,k}(\mathbf{q}_2) \triangleq  2^{\frac{R_0}{W M_k}  \big( 1 + \frac{1.28\lambda_u \mathcal{P}_{k}(\mathbf{q}_2)}{\lambda_k}\big) } - 1$. Then, by substituting \eqref{eqn:suc2} into \eqref{eqn:off}, Proposition 1 can be proved.\qed

\section{Proof of Corollary 1}
\renewcommand{\theequation}{C.\arabic{equation}}
\setcounter{equation}{0}
Without loss of generality, we assume $\ubar q_{1,2}^*, \dots, \ubar q_{N_1,2}^* = 1$ and $\ubar q_{N_f-N_0+1,2}^*, \dots, \ubar q_{N_f,2}^* = 0$. By defining $ c_1 \triangleq  \mathsf{C}_{1}(\bar\gamma_{0,2})  + \bar{p}_{\rm a, 2}\mathsf{C}_{2}(\bar\gamma_{0,2})$ and $ c_2 \triangleq \bar p_{\rm a, 2}C_{3,\bar\gamma_{0,2}} + 1$, from $\sum_{f=1}^{N_f} \ubar q_{f,2}^* = N_c$ and \eqref{eqn:optlower},  we have $\sum_{f=N_0 + 1}^{N_f-N_1} \left( \frac{1}{ c_2} \sqrt{\frac{c_1}{\nu}} \sqrt{p_f} - \frac{c_1}{c_2} \right) + N_1 = N_c$, from which we obtain
\begin{equation}
\frac{1}{c_2}\sqrt{\frac{c_1}{\nu}} = \frac{{N_c - N_1 + (N_f - N_0 - N_1)\frac{c_1}{c_2}}}{ \sum_{f=N_0 + 1}^{N_f-N_1} \sqrt{p_f}}. \label{eqn:c1}
\end{equation}
As shown in \eqref{eqn:c1}, $\frac{1}{c_2}\sqrt{\frac{c_1}{\nu}}$ increases with ${c_1}/{c_2}$. Since $\mathsf{C}_{2}(\bar\gamma_{0,2}) \geq 0$ and $\mathsf{C}_{3}(\bar\gamma_{0,2}) \leq 0$, ${c_1}/{c_2}$ increases with $\bar p_{{\rm a}, 2}$ and hence $\frac{1}{c_2}\sqrt{\frac{c_1}{\nu}}$ increases with $\bar p_{{\rm a}, 2}$. Further considering that $\bar p_{{\rm a},2}$ increases with $\lambda_u/\lambda_2$, we know that $\frac{1}{c_2}\sqrt{\frac{c_1}{\nu}}$ increases with  $\lambda_u/\lambda_2$.

Then, for any $\ubar q_{f,2}^*$, $\ubar q_{f+1,2}^* \in (0,1)$, from \eqref{eqn:optlower}, we have $
\ubar q_{f,2}^* - \ubar q_{f+1,2}^* =  \frac{1}{c_2}\sqrt{\frac{c_1}{\nu}} (p_f - p_{f+1})$. Since $p_f > p_{f+1}$, $\ubar q_{f,2}^*-\ubar q_{f+1,2}^* $ increases with $\frac{1}{c_2}\sqrt{\frac{c_1}{\nu}}$ and hence increases with  $\lambda_u/\lambda_2$.~\qed

\section{Proof of Corollary 2}

\renewcommand{\theequation}{D.\arabic{equation}}
\setcounter{equation}{0}
Since $R_0 \to \infty$ or $\lambda_u/\lambda_2 \to \infty$ are equivalent to $\bar \gamma_{0,2} \to \infty$, in the following we analyze the case for $\bar \gamma_{0,2} \to \infty$ instead. By employing the transformation  ${}_2F_1[a,b;c;x]  =  \tfrac{\Gamma(c)\Gamma(b-a)}{\Gamma(b)\Gamma(c-a)} (-x)^{-a}{}_2F_1[a,a+1-c;a+1$ $-b;\tfrac{1}{x}]
+  \tfrac{\Gamma(c)\Gamma(a-b)}{\Gamma(a)\Gamma(c-b)} (-x)^{-b}{}_2F_1[b,b+1-c;b+1-a;\tfrac{1}{x}] $ \cite[eq. (9.132)]{jeffrey2007table}
and considering the series-form expression of ${}_2F_1 [a,b;c;x] = \sum_{n = 0}^{\infty} \frac{(a)_n 	(b)_n}{(c)_n} x^n$ where $(a)_n \triangleq a(a+1) \cdots (a + n - 1)$
denotes the rising Pochhammer symbol, we can obtain the asymptotic result of $ {}_{2}F_1  \big[ -\frac{2}{\alpha}, M_2; 1-\frac{2}{\alpha}; -\bar \gamma_{0,2} \big]$ for $\bar \gamma_{0,2} \to \infty$ as,
\begin{align}
{}_{2}F_1  \left[ -\frac{2}{\alpha}, M_2; 1-\frac{2}{\alpha}; -\bar\gamma_{0,2} \right]= &\frac{\Gamma (1-\frac{2}{\alpha}) \Gamma(M_2 + \frac{2}{\alpha})}{\Gamma(M_2)} \bar\gamma_{0,2}{}^{\frac{2}{\alpha}} \nonumber\\
&
+ \mathcal{O}\left(\bar \gamma_{0,2}^{-M_2}\right), \label{eqn:asy}
\end{align}
which equals $\mathsf{C}_{2,\bar \gamma_{0,2}}$ when $\bar \gamma_{0,2} \to \infty$. Then, considering the definition of $\mathsf{C}_{3, \bar \gamma_{0,2}}$, we have $\lim\limits_{\bar \gamma_{0,2} \to \infty}\mathsf{C}_{3,\bar \gamma_{0,2}} = \lim\limits_{\bar \gamma_{0,2} \to \infty}{}_{2}F_1  \big[ -\frac{2}{\alpha}, M_2; 1-\frac{2}{\alpha}; -\bar\gamma_{0,2} \big] - \mathsf{C}_{2}(\bar\gamma_{0,2})- 1 = -1$.  Upon substituting $\mathsf{C}_{3}(\bar\gamma_{0,2}) = -1$ into \eqref{eqn:case1}, we obtain $\lim\limits_{R_0 \to \infty} p_{\rm s}(\mathbf{q}_2) = \lim\limits_{\bar \gamma_{0,2} \to \infty} p_{\rm s}(\mathbf{q}_2) = \frac{1}{\mathsf{C}_{1,\bar \gamma_{0,2}}}\sum_{f=1}^{N_f} p_fq_{f,2}$. Since $p_f$ decreases with $f$ and further considering constraint \eqref{eqn:con1} and \eqref{eqn:con2}, it is easy to see that the optimal values of $q_{f,2}$ maximizing $\sum_{f = 1}^{N_f} p_f q_{f,2}$ are $\ubar q_{1,2}^*, \cdots, \ubar q_{N_c,2}^* = 1$ and $\ubar q_{N_c+1,2}^*, \cdots, \ubar q_{N_f, 2}^* = 0$, and hence Corollary 2 is proved. \qed

\section{Proof of Corollary 3}

\renewcommand{\theequation}{E.\arabic{equation}}
\setcounter{equation}{0}
Without loss of generality, we assume $q_{1,2}^*, \dots, \ubar q_{N_1,2}^* = 1$ and $q_{N_f-N_0+1,2}^*, \dots, q_{N_f,2}^* = 0$, from $\sum_{f=1}^{N_f} q_{f,2}^* = N_c$ and \eqref{eqn:opt}, similar to the derivation of \eqref{eqn:c1}, we can obtain
\begin{equation}
\sqrt{\frac{\mathsf{C}_{1}(\gamma_{0,2})}{\nu}}  = \frac{N_c - N_1 + (N_f - N_0 - N_1)\mathsf{C}_{1}(\gamma_{0,2})}{ \sum_{f=N_0 + 1}^{N_f-N_1} \sqrt{p_f}}. \label{eqn:k}
\end{equation}
As shown in \eqref{eqn:k}, $\sqrt{\frac{\mathsf{C}_{1}(\gamma_{0,2})}{\nu}}$ increases with $\mathsf{C}_{1}(\gamma_{0,2})$. Considering that $\mathsf{C}_{1}(\gamma_{0,2})$ increases with $\lambda_{12}$ and $P_{12}$, $\sqrt{\frac{\mathsf{C}_{1}(\gamma_{0,2})}{\nu}}$ increases with $\lambda_{12}$ and $P_{12}$.

Then, for any $q_{f,2}^*$, $q_{f+1,2}^* \in (0,1)$, from \eqref{eqn:optlower}, we have
$q_{f,2}^* - q_{f+1,2}^* =  \sqrt{\frac{\mathsf{C}_{1}(\gamma_{0,2})}{\nu}} (p_f - p_{f+1})$ and $q_{f,2}^*-q_{f+1,2}^* $ increases with $\sqrt{\frac{\mathsf{C}_{1}(\gamma_{0,2})}{\nu}}$ and hence increases with $\lambda_{12}$ and $P_{12}$. \qed

\section{ Proof of Proposition 3}

\renewcommand{\theequation}{F.\arabic{equation}}
\setcounter{equation}{0}
The average throughput of an active BS in the $k$th tier can be expressed as
\begin{align}
\mathbb{E}[R_k(\mathbf{q}_2)] &= \mathbb{E} \left[\sum_{u=1}^{U_k}  \frac{W}{U_k/M_k} \ln \left( 1 + \gamma_{uk}(\mathbf{q}_2)\right) \right] \nonumber \\
& = W M_k \mathbb{E} \left[ \ln \left( 1 + \gamma_{uk}(\mathbf{q}_2)\right) \right], \label{eqn:ERk}
\end{align}
where the last step follows because $\gamma_{uk}(\mathbf{q}_2)~(u = 1,\cdots,U_k)$ are independently and identically distributed.
By taking the expectation over the random file request, we have
\begin{align}
\mathbb{E} [ \ln \left( 1 + \gamma_{uk}(\mathbf{q}_2)\right) ] = \sum_{f=1}^{N_f} p_{f,k}(\mathbf{q}_2)\mathbb{E} [ \ln(1+ \gamma_{f,k}(\mathbf{q}_2)) ] ,   \label{eqn:Eln}
\end{align}
where $p_{f,k}(\mathbf{q}_2) = \frac{p_f \mathcal{P}_{f,k}(\mathbf{q}_2)}{\mathcal{P}_k(\mathbf{q}_2)}$ obtained from conditional probability formula is the probability that the $u$th user served by the BS in the $k$th tier requests the $f$th file. Since $\mathbb{E}[X] = \int_{0}^{\infty} \mathbb{P}(X>x) {\rm d} x$ for $X > 0$, $\mathbb{E} \left[ \ln \left( 1 + \gamma_{f,k}(\mathbf{q}_2)\right) \right]$ can be derived as
\begin{align}
\mathbb{E} \left[ \ln \left( 1 + \gamma_{f,k}(\mathbf{q}_2)\right) \right] & =\int_{0}^{\infty} \mathbb{P} (\ln \left( 1 + \gamma_{f,k}(\mathbf{q}_2)\right) > x) {\rm d} x \nonumber \\
& = \int_{0}^{\infty}  \mathbb{P} ( \gamma_{f,k}(\mathbf{q}_2) > e^x - 1) {\rm d} x .\label{eqn:rate}
\end{align}
Similar to deriving \eqref{eqn:suc}, we can obtain
\begin{align}
&\mathbb{P} ( \gamma_{f,k}(\mathbf{q}_2) > e^x - 1) \nonumber \\
  = & \frac{q_{f,k}}{\mathcal{P}_{f,k}(\mathbf{q}_2)} \bigg(\sum_{j=1}^{2}  \lambda_{jk} P_{jk}^{\frac{2}{\alpha}} \Big(q_{f,j}  p_{{\rm a},j}(\mathbf{q}_2) B_{jk}^{\frac{2}{\alpha}} \mathsf{Z}_{1,jk}(e^x-1) \nonumber \\
& + (1- q_{f,j})  p_{{\rm a},j}(\mathbf{q}_2) \mathsf{Z}_{2,jk}(e^x-1)  +  q_{f,j}   B_{jk}^{\frac{2}{\alpha}} \Big) \bigg)^{-1}. \label{eqn:succ}
\end{align}
Upon substituting \eqref{eqn:succ} into \eqref{eqn:rate} and furthering considering \eqref{eqn:Eln}, \eqref{eqn:ERk} and \eqref{eqn:ASEdef}, Proposition 3 can be proved. \qed

\section{Proof of Corollary 4}

\renewcommand{\theequation}{G.\arabic{equation}}
\setcounter{equation}{0}

From the series-form expression of ${}_2F_1[a,b;c;x] = \sum_{n = 0}^{\infty} \frac{(a)_n (b)_n}{(c)_n} x^n$, we have the asymptotic expression of ${}_2F_1[a,b;c;x]$ for $x\to 0$ as ${}_2F_1[a,b;c;x] = 1 + \frac{ab}{c}x + \mathcal{O}(x^2)$ which is accurate when $x\ll 1$. Then, the asymptotic expression of $\mathsf{Z}_{1,jk}(e^x - 1)$ for $x \to 0$ can be derived as
\begin{align}
\mathsf{Z}_{1,jk}(e^x - 1) & = {}_{2}F_1 \left[ -\frac{2}{\alpha}, M_j; 1-\frac{2}{\alpha}; -\frac{e^x-1}{M_{jk}B_{jk}} \right] -1 \nonumber \\
& = 1 + \frac{2M_k}{2-\alpha} (1 - e^x) -1 + \mathcal{O}((1-e^x)^2) \nonumber \\
&= \frac{2M_k}{2-\alpha} x +   \mathcal{O}(x^2),  \label{eqn:ap1}
\end{align}
which is accurate when $1 - e^x \ll 1$, i.e., $x \ll \ln2$. In this case, the asymptotic expression of $\mathsf{Z}_{2,jk}(e^x - 1)$ for $x \to 0$ can be expressed as
\begin{align}
\!\! \mathsf{Z}_{2,jk}(e^x - 1) &= \frac{\Gamma \left(1-\frac{2}{\alpha}\right) \Gamma\left(M_j + \frac{2}{\alpha}\right) }{\Gamma(M_j)} \left(\frac{e^x - 1}{M_{jk}}\right)^{\frac{2}{\alpha}} \nonumber \\
 &=\frac{\Gamma \left(1-\frac{2}{\alpha}\right) \Gamma\left(M_j + \frac{2}{\alpha}\right) }{\Gamma(M_j)M_{jk}{}^{\frac{2}{\alpha}}} x^{\frac{2}{\alpha}} +  \mathcal{O}(x^{\frac{4}{\alpha}}) .\label{eqn:ap2}
\end{align}

When $x \to \infty$, similar to the derivation of \eqref{eqn:asy}, we can obtain the asymptotic expression of $\mathsf{Z}_{1,fk}(e^x - 1)$ for $x\to \infty$ as,
\begin{align}
& \mathsf{Z}_{1,jk}(e^x - 1) \nonumber \\
& = \frac{\Gamma(1 - \frac{2}{\alpha}) \Gamma(M_j + \frac{2}{\alpha})}{\Gamma(M_j) M_{jk}{}^{\frac{2}{\alpha}}} \left(e^x - 1\right)^{\frac{2}{\alpha}} - 1 + \mathcal{O}((e^x - 1)^{-M_j})  \nonumber \\
& \approx \frac{\Gamma(1 - \frac{2}{\alpha}) \Gamma(M_j + \frac{2}{\alpha})}{\Gamma(M_j) M_{jk}{}^{\frac{2}{\alpha}}} e^{\frac{2x}{\alpha}} - 1 +\mathcal{O}(x^{-M_j}) , \label{eqn:ap3}
\end{align}
where the last step omits ``1" in $e^x - 1$ which is accurate when $1-e^x \gg 1$, i.e., $x \gg \ln 2$. Then, the asymptotic expression of $\mathsf{Z}_{2,jk}(e^x - 1)$ for $x \to \infty$ can be expressed as
\begin{equation}
\mathsf{Z}_{2,jk}(e^x - 1) \approx \frac{\Gamma(1 - \frac{2}{\alpha}) \Gamma(M_j + \frac{2}{\alpha})}{\Gamma(M_j) M_{jk}{}^{\frac{2}{\alpha}}} e^{\frac{2x}{\alpha}}. \label{eqn:ap4}
\end{equation}

By dividing the integration limits from $[0, \infty)$ to $[0, \ln 2]$ and $[\ln 2, \infty)$, we have
\begin{align}
\mathsf{Int}_{fk}(x)|_{0}^{\infty} = \mathsf{Int}_{fk}(x)|_{0}^{\ln 2}  + \mathsf{Int}_{fk}(x)|_{\ln 2}^{\infty}. \label{eqn:int}
\end{align}

Employing the asymptotic expressions \eqref{eqn:ap1} and \eqref{eqn:ap2} for $\mathsf{Z}_{1,jk}(e^x -1)$ and $\mathsf{Z}_{2,jk}(e^x -1)$ and omitting $\mathcal{O}(x^2)$ when $x \in [0,\ln 2]$, we can approximate $\mathsf{Int}_{fk}(x)|_{0}^{\ln2}$ as
\begin{align}
\mathsf{Int}_{fk}(x)|_{0}^{\ln 2} 
\approx & \int_{0}^{\ln 2}  \left(\sum_{j=1}^{2}  \lambda_{jk} P_{jk}^{\frac{2}{\alpha}} \left(q_{f,j}  p_{{\rm a},j}(\mathbf{q}_2) \frac{2M_k}{2-\alpha} x \right. \right.\nonumber \\
& + ( 1 -  q_{f,j})  p_{{\rm a},j}(\mathbf{q}_2) \frac{\Gamma(1 - \frac{2}{\alpha}) \Gamma(M_j + \frac{2}{\alpha})}{\Gamma(M_j) M_{jk}{}^{\frac{2}{\alpha}}} x^{\frac{2}{\alpha}} \nonumber \\
& \left. +  q_{f,j}    \bigg)\vphantom{\sum_{j=1}^{2}}\right)^{-1}  {\rm d} x  \nonumber \\
 \approx  &\int_{0}^{\ln 2} \left(\sum_{j=1}^{2}   \lambda_{jk} P_{jk}^{\frac{2}{\alpha}}q_{f,j} \right)^{-1} {\rm d} x \nonumber \\
= & \left(\sum_{j=1}^{2} q_{f,j} \lambda_{jk}  P_{jk}^{\frac{2}{\alpha}} \right)^{-1} \ln 2,   \label{eqn:int1}
\end{align}
where the last approximation is from $\mathbf{q}_1 = \mathbf{1}$ and the condition that $N_c/N_f$ is large so that most of files can be cached at the helper, i.e., $\mathbf{q}_{2} \to \mathbf{1}$, and hence we can omit the term $q_{f,j}  p_{{\rm a},j}(\mathbf{q}_2) \tfrac{2M_k}{2-\alpha} x +( 1 -  q_{f,j})  p_{{\rm a},j}(\mathbf{q}_2) \frac{\Gamma(1 - \frac{2}{\alpha}) \Gamma(M_j + \frac{2}{\alpha})}{\Gamma(M_j) M_{jk}{}^{\frac{2}{\alpha}}} x^{\frac{2}{\alpha}} $ when compared with $q_{f,j}$ for $x\in [0, \ln 2]$.

Employing the asymptotic expressions in \eqref{eqn:ap3}, \eqref{eqn:ap4} and omitting $\mathcal{O}(x^{-M_j})$  when $x \in [\ln 2,\infty]$, we obtain
\begin{align}
\mathsf{Int}_{fk}(x)|_{\ln 2}^{\infty} \approx& \int_{\ln 2}^{\infty}  \left(  \sum_{j=1}^{2}  \lambda_{jk} P_{jk}^{\frac{2}{\alpha}}  p_{{\rm a},j}(\mathbf{q}_2) \right. \nonumber \\
&\times \frac{\Gamma(1 - \frac{2}{\alpha}) \Gamma(M_j + \frac{2}{\alpha})}{\Gamma(M_j) M_{jk}{}^{\frac{2}{\alpha}}} e^{\frac{2x}{\alpha}} \nonumber \\
&\left. + \sum_{j=1}^{2}  \lambda_{jk} P_{jk}^{\frac{2}{\alpha}} q_{f,j} (1 - p_{{\rm a},j}(\mathbf{q}_2)) \vphantom{\sum_{j=1}^{2}} \right)^{-1} {\rm d} x \nonumber \\
= & \frac{\alpha}{2\mathsf{K}_{2,fk}(\mathbf{q}_2)} \ln \left( 1 + \frac{\mathsf{K}_{2,fk}(\mathbf{q}_2)}{\mathsf{K}_{1,k}(\mathbf{q}_2)} 4^{-\frac{1}{\alpha}}\right), \label{eqn:int2}
\end{align}
where $\mathsf{K}_{1,k}(\mathbf{q}_2) \triangleq  \sum_{j=1}^{2}  \lambda_{jk} P_{jk}^{\frac{2}{\alpha}}  p_{{\rm a},j}(\mathbf{q}_2) \Gamma(1 - \frac{2}{\alpha}) \Gamma(M_j + \frac{2}{\alpha})\Gamma(M_j)^{-1} M_{jk}^{-\frac{2}{\alpha}}$, and $\mathsf{K}_{2,fk}(\mathbf{q}_2) \triangleq \sum_{j=1}^{2}  \lambda_{jk} P_{jk}^{\frac{2}{\alpha}} q_{f,j}(1 - p_{{\rm a},j}(\mathbf{q}_2))$.

Upon substituting \eqref{eqn:int1} and \eqref{eqn:int2} into \eqref{eqn:int} and further considering \eqref{eqn:ASE} and the expression of $\mathcal{P}_k$ given in Lemma 1, we can obtain Corollary 4. \qed

\section{Proof of Corollary 5}

\renewcommand{\theequation}{H.\arabic{equation}}
\setcounter{equation}{0}
With $\lambda_u/\lambda_2 \to \infty$, we have $p_{{\rm a},k}(\mathbf{q}_2) \to 1$ and $\mathsf{K}_{2,fk}(\mathbf{q}_2) \to 0$.  Then, \eqref{eqn:ASEc} degenerates into,

\begin{equation}
{\sf ASE}(\mathbf{q}_2) \approx  \sum_{k=1}^{2}  \lambda_k  M_k \left(\ln 2 +  \frac{\alpha}{2 \mathsf{K}_{1,k}\mathcal{P}_{k}(\mathbf{q}_2)}   4^{-\frac{1}{\alpha}} \sum_{f=1}^{N_f}  p_f q_{f,k} \right). \label{eqn:ASEs} 
\end{equation}
where $\mathsf{K}_{1,k}$ does not depend on $q_{f,k}$ since $p_{{\rm a}, k}(\mathbf{q}_2) \to 1$.

Further considering $\lambda_2/\lambda_1 \to \infty$, we have $\mathcal{P}_2(\mathbf{q}_2) \to 1$ and the ASE of the MBS tier can be omitted compared with the ASE of the helper tier,

\begin{equation}
{\sf ASE}(\mathbf{q}_2) \approx \lambda_2 M_2 \left( \ln 2 +\frac{\alpha}{2 \mathsf{K}_{1,2}}   4^{-\frac{1}{\alpha}} \sum_{f=1}^{N_f}  p_f q_{f,2} \right). \label{eqn:dASE} 
\end{equation}
From \eqref{eqn:dASE} we can see that maximizing ASE is equivalent to maximizing $\sum_{f=1}^{N_f}  p_f q_{f,2}$. Further considering constraints \eqref{eqn:con1} and \eqref{eqn:con2}, we can easily obtain the optimal caching probability as $ q_{1,2}^*, \cdots,  q_{N_c,2}^* = 1$ and $ q_{N_c+1,2}^*, \cdots,  q_{N_f, 2}^* = 0$.  \qed

\section{Proof of Proposition 4}

\renewcommand{\theequation}{I.\arabic{equation}}
\setcounter{equation}{0}
For the PBS tier, based on the law of total probability, from \eqref{eqn:spER2} we obtain

\begin{align}
&\mathbb{E}[R_2]  = \sum_{f=1}^{N_f} p_{f,2}(\mathbf{q}_2)\mathbb{E} [ \min\{ W\ln(1 + \gamma_{f,2}(\mathbf{q}_2)), C_{\rm bh, 2} \}] \nonumber \\
& = W \sum_{f=1}^{N_f} p_{f,2}(\mathbf{q}_2)\mathbb{E} \left[ \min \left\{ \ln(1 + \gamma_{f,2}(\mathbf{q}_2)), \frac{C_{\rm bh, 2}}{W} \right\}\right] \nonumber \\
&= W \sum_{f=1}^{N_f} p_{f,2}(\mathbf{q}_2) \int_{0}^{\infty} \mathbb{P}\left[ \min \left\{ \ln(1 + \gamma_{f,2}(\mathbf{q}_2)),  \frac{C_{\rm bh, 2}}{W} \right\} > x\right] {\rm d} x \nonumber \\
& = W \sum_{f=1}^{N_f} p_{f,2}(\mathbf{q}_2) \int_{0}^{\frac{C_{\rm bh, 2}}{W}} \mathbb{P}[ \log_2(1 + \gamma_{f,2}(\mathbf{q}_2)) > x] {\rm d} x ,  \label{eqn:ER2}
\end{align}
where the last step is from $ \mathbb{P}[\min\{ \log_2(1 + \gamma_{f,2}(\mathbf{q}_2)), \frac{C_{\rm bh, 2}}{W} \} > x] =0$ for $x>\frac{C_{\rm bh, 2}}{W}$ and $ \mathbb{P}[\min\{ \log_2(1 + \gamma_{f,2}(\mathbf{q}_2)), \frac{C_{\rm bh, 2}}{W} \} > x] =\mathbb{P}[\log_2(1 + \gamma_{f,2}(\mathbf{q}_2)) > x]$ for $0\leq x \leq \frac{C_{\rm bh, 2}}{W}$.

Since the user association of traditional HetNet can be regarded as a special case of cache-enabled HetNets when $\mathbf{q}_1 = \mathbf{q}_2 = \mathbf{1}$,
by substituting \eqref{eqn:succ} with $k=2$ and $\mathbf{q}_1 = \mathbf{q}_2 = \mathbf{1}$ into \eqref{eqn:ER2}, we have
\begin{align}
\mathbb{E}[R_2] = & \frac{W}{\mathcal{P}_2} \int_{0}^{\frac{C_{\rm bh, 2}}{W}}  \left(\sum_{j=1}^{2}  \lambda_{j2} (P_{j2}B_{j2})^{\frac{2}{\alpha}}  \right. \nonumber \\
&  \left.\times(  p_{{\rm a},j} \mathsf{Z}_{1,j2}(e^x-1) + 1 ) \vphantom{\sum_{j=1}^{2}}\right)^{-1}   {\rm d} x.
\end{align}
where $\mathcal{P}_2 = \mathcal{P}_2(\mathbf{q}_2)|_{\mathbf{q}_2 = \mathbf{1}}$.

Since the deployment of MBSs in the traditional HetNet are the same as those in the cache-enabled HetNet, by substituting \eqref{eqn:succ} with $k=1$ and $\mathbf{q}_1 = \mathbf{q}_2 = \mathbf{1}$ into \eqref{eqn:rate} and furthering considering \eqref{eqn:Eln} and \eqref{eqn:ERk}, we can obtain
\begin{align}
\mathbb{E}[R_1] =& \frac{WM_1}{\mathcal{P}_1}\int_{0}^{\infty}  \Bigg(\sum_{j=1}^{2}  \lambda_{j1} (P_{j1}B_{j1})^{\frac{2}{\alpha}} \nonumber \\
&\times \left(  p_{{\rm a},j} \mathsf{Z}_{1,j1}(e^x-1)  + 1 \right) \Bigg)^{-1} {\rm d} x. \label{eqn:ER1}
\end{align}
where $\mathcal{P}_1 = \mathcal{P}_1(\mathbf{q}_2)|_{\mathbf{q}_2 = \mathbf{1}}$. Upon substituting \eqref{eqn:ER1} and \eqref{eqn:ER2} into \eqref{eqn:ASEdef}, Proposition 4 can be proved. \qed

\section{Proof of Corollary 6}
\renewcommand{\theequation}{J.\arabic{equation}}
\setcounter{equation}{0}
Similar to the derivation of Colloary 4, by approximating $\mathsf{Z}_{1,jk}(e^x - 1)$ as \eqref{eqn:ap1} for $x \in [0,\ln 2)$ and as \eqref{eqn:ap2} for $x \in [\ln 2, \infty)$, \eqref{eqn:ER1} and \eqref{eqn:ER2} can be approximated as
\begin{align}
\mathbb{E}[R_2] \approx& \frac{W}{\mathcal{P}_2} \int_{0}^{\frac{C_{\rm bh, 2}}{W}}  \left(\sum_{j=1}^{2}  \lambda_{j2} P_{j2}^{\frac{2}{\alpha}}  \right)^{-1}  {\rm d} x \nonumber \\
= &\frac{C_{\rm bh, 2}}{\mathcal{P}_2 \sum_{j=1}^{2}  \lambda_{j2} P_{j2}^{\frac{2}{\alpha}}} = C_{\rm bh, 2},  \label{eqn:ER2b} \\
\mathbb{E}[R_1] \approx &\frac{WM_1}{\mathcal{P}_1} \left(\int_{0}^{\ln 2}  \left(\sum_{j=1}^{2}  \lambda_{j1} P_{j1}^{\frac{2}{\alpha}}  \right)^{-1}  {\rm d} x \right. \nonumber \\
& + \int_{\ln 2}^{\infty}  \left(\sum_{j=1}^{2}  \lambda_{j1} P_{j1}^{\frac{2}{\alpha}} \left( p_{{\rm a},j} \frac{\Gamma(1 - \frac{2}{\alpha}) \Gamma(M_j + \frac{2}{\alpha})}{\Gamma(M_j) M_{j1}^{\frac{2}{\alpha}}} e^{\frac{2x}{\alpha}} \right. \right. \nonumber \\
& \left.\left.\left.  + 1 - p_{{\rm a}, j} \vphantom{\frac{\Gamma(1 - \frac{2}{\alpha}) \Gamma(M_j + \frac{2}{\alpha})}{\Gamma(M_j) M_{j1}^{\frac{2}{\alpha}}}}\right)\right)^{-1}{\rm d} x \vphantom{\left(\sum_{j=1}^{2}  \lambda_{j1} P_{j1}^{\frac{2}{\alpha}}  \right)^{-1}} \right)   \nonumber \\
= & WM_1 \left(\ln 2 +\frac{\alpha}{2\mathcal{P}_1 \mathsf{K}_{2}}  \ln \left(1 + \frac{\mathsf{K}_2}{\mathsf{K}_1}4^{-\frac{1}{\alpha}}\right)   \right), \label{eqn:ER1b}
\end{align}
where $\mathsf{K}_{1} \triangleq  \sum_{j=1}^{2}  \lambda_{j1} P_{j1}^{\frac{2}{\alpha}}  p_{{\rm a},j} \Gamma(1 - \frac{2}{\alpha}) \Gamma(M_j + \frac{2}{\alpha})\Gamma(M_j)^{-1}$ $ M_{j1}^{-\frac{2}{\alpha}}$, and $\mathsf{K}_{2} \triangleq \sum_{j=1}^{2}  \lambda_{j1} P_{j1}^{\frac{2}{\alpha}} (1 - p_{{\rm a},j})$. Upon substituting \eqref{eqn:ER1b} and \eqref{eqn:ER2b} into \eqref{eqn:ASEdef}, Corollary 5 can be proved. \qed

\bibliographystyle{IEEEtran}
\bibliography{dongbib}

\begin{thebibliography}{10}
\providecommand{\url}[1]{#1}
\csname url@samestyle\endcsname
\providecommand{\newblock}{\relax}
\providecommand{\bibinfo}[2]{#2}
\providecommand{\BIBentrySTDinterwordspacing}{\spaceskip=0pt\relax}
\providecommand{\BIBentryALTinterwordstretchfactor}{4}
\providecommand{\BIBentryALTinterwordspacing}{\spaceskip=\fontdimen2\font plus
\BIBentryALTinterwordstretchfactor\fontdimen3\font minus
  \fontdimen4\font\relax}
\providecommand{\BIBforeignlanguage}[2]{{%
\expandafter\ifx\csname l@#1\endcsname\relax
\typeout{** WARNING: IEEEtran.bst: No hyphenation pattern has been}%
\typeout{** loaded for the language `#1'. Using the pattern for}%
\typeout{** the default language instead.}%
\else
\language=\csname l@#1\endcsname
\fi
#2}}
\providecommand{\BIBdecl}{\relax}
\BIBdecl

\bibitem{DongICC}
D.~Liu and C.~Yang, ``Cache-enabled heterogeneous cellular networks: Comparison
  and tradeoffs,'' in \emph{Proc. IEEE ICC}, 2016.

\bibitem{DongGC}
------, ``Optimal content placement for offloading in cache-enabled
  heterogeneous wireless networks,'' in \emph{Proc. IEEE GLOBECOM}, 2016.

\bibitem{densification}
N.~Bhushan, J.~Li, D.~Malladi, R.~Gilmore, D.~Brenner, A.~Damnjanovic, R.~T.
  Sukhavasi, C.~Patel, and S.~Geirhofer, ``Network densification: the dominant
  theme for wireless evolution into {5G},'' \emph{IEEE Commun. Mag.}, vol.~52,
  no.~2, pp. 82--89, Feb. 2014.

\bibitem{Femtocell}
V.~Chandrasekhar, J.~Andrews, and A.~Gatherer, ``Femtocell networks: a
  survey,'' \emph{IEEE Commun. Mag.}, vol.~46, no.~9, pp. 59--67, Sept. 2008.

\bibitem{woo2013comparison}
S.~Woo, E.~Jeong, S.~Park, J.~Lee, S.~Ihm, and K.~Park, ``Comparison of caching
  strategies in modern cellular backhaul networks,'' in \emph{Proc. ACM
  MobiSys}, 2013.

\bibitem{Andy2012}
N.~Golrezaei, K.~Shanmugam, A.~G. Dimakis, A.~F. Molisch, and G.~Caire,
  ``Femtocaching: Wireless video content delivery through distributed caching
  helpers,'' in \emph{Proc. IEEE INFOCOM}, 2012.

\bibitem{Procach14}
E.~Bastug, M.~Bennis, and M.~Debbah, ``Living on the edge: The role of
  proactive caching in 5{G} wireless networks,'' \emph{{IEEE} Commun. Mag.},
  vol.~52, no.~8, pp. 82--89, Aug. 2014.

\bibitem{Dong}
D.~Liu and C.~Yang, ``Energy efficiency of downlink networks with caching at
  base stations,'' \emph{IEEE J. Sel. Areas Commun.}, vol.~34, no.~4, pp.
  907--922, Apr. 2016.

\bibitem{flexible}
H.-S. Jo, Y.~J. Sang, P.~Xia, and J.~Andrews, ``Heterogeneous cellular networks
  with flexible cell association: A comprehensive downlink {SINR} analysis,''
  \emph{IEEE Trans. Wireless Commun.}, vol.~11, no.~10, pp. 3484--3495, Oct.
  2012.

\bibitem{EURASIP}
E.~Bastug, M.~Bennis, M.~Kountouris, and M.~Debbah, ``Cache-enabled small cell
  networks: modeling and tradeoffs,'' \emph{EURASIP J. on Wireless Commun. and
  Netw.}, vol. 2015, no.~1, 2015.

\bibitem{BastugISIT}
------, ``Edge caching for coverage and capacity-aided heterogeneous
  networks,'' in \emph{Proc. ISIT}, 2016.

\bibitem{chenchenyang}
C.~Yang, Y.~Yao, Z.~Chen, and B.~Xia, ``Analysis on cache-enabled wireless
  heterogeneous networks,'' \emph{IEEE Tran. Wireless Commun.}, vol.~15, no.~1,
  pp. 131--145, Jan. 2016.

\bibitem{DongMag}
D.~Liu, B.~Chen, C.~Yang, and A.~F. Molisch, ``Caching at the wireless edge:
  design aspects, challenges, and future directions,'' \emph{IEEE Commun.
  Mag.}, vol.~54, no.~9, pp. 22--28, Sept. 2016.

\bibitem{Niki13}
N.~Golrezaei, A.~F. Molisch, A.~G. Dimakis, and G.~Caire, ``Femtocaching and
  device-to-device collaboration: A new architecture for wireless video
  distribution,'' \emph{IEEE Commun. Mag.}, vol.~51, no.~4, pp. 142--149, Apr.
  2013.

\bibitem{song2015optimal}
J.~Song, H.~Song, and W.~Choi, ``Optimal caching placement of caching system
  with helpers,'' in \emph{proc. IEEE ICC}, 2015.

\bibitem{Blaszczyszyn2015optimal}
B.~Blaszczyszyn and A.~Giovanidis, ``Optimal geographic caching in cellular
  networks,'' in \emph{proc. IEEE ICC}, 2015.

\bibitem{cui2015analysis}
Y.~Cui, D.~Jiang, and Y.~Wu, ``Analysis and optimization of caching and
  multicasting in cache-enabled wireless networks,'' \emph{IEEE Trans. Wireless
  Commun., early access}, 2016.

\bibitem{rao2015optimal}
J.~Rao, H.~Feng, C.~Yang, Z.~Chen, and B.~Xia, ``Optimal caching placement for
  {D2D} assisted wireless caching networks,'' in \emph{Proc. IEEE ICC}, 2016.

\bibitem{breslau1999web}
L.~Breslau, P.~Cao, L.~Fan, G.~Phillips, and S.~Shenker, ``Web caching and
  {Zipf}-like distributions: Evidence and implications,'' in \emph{Proc. IEEE
  INFOCOM}, 1999.

\bibitem{zhang2011multi}
J.~Zhang, M.~Kountouris, J.~G. Andrews, and R.~W. Heath, ``Multi-mode
  transmission for the {MIMO} broadcast channel with imperfect channel state
  information,'' \emph{IEEE Trans. Commun.}, vol.~59, no.~3, pp. 803--814, Mar.
  2011.

\bibitem{adhoc}
N.~Jindal, J.~Andrews, and S.~Weber, ``Multi-antenna communication in ad hoc
  networks: Achieving {MIMO} gains with {SIMO} transmission,'' \emph{IEEE
  Trans. Commun.}, vol.~59, no.~2, pp. 529--540, Feb. 2011.

\bibitem{Quek}
Y.~S. Soh, T.~Quek, M.~Kountouris, and H.~Shin, ``Energy efficient
  heterogeneous cellular networks,'' \emph{IEEE J. Sel. Areas Commun.},
  vol.~31, no.~5, pp. 840--850, May 2013.

\bibitem{economy}
S.~Lee and K.~Huang, ``Coverage and economy of cellular networks with many base
  stations,'' \emph{IEEE Commun. Lett.}, vol.~16, no.~7, pp. 1038--1040, July
  2012.

\bibitem{offloading}
S.~Singh, H.~Dhillon, and J.~Andrews, ``Offloading in heterogeneous networks:
  Modeling, analysis, and design insights,'' \emph{IEEE Trans. Wireless
  Commun.}, vol.~12, no.~5, pp. 2484--2497, May 2013.

\bibitem{boyd2004convex}
S.~Boyd and L.~Vandenberghe, \emph{Convex optimization}.\hskip 1em plus 0.5em
  minus 0.4em\relax Cambridge university press, 2004.

\bibitem{dahrouj2015cost}
H.~Dahrouj, A.~Douik, F.~Rayal, T.~Y. Al-Naffouri, and M.-S. Alouini,
  ``Cost-effective hybrid {RF/FSO} backhaul solution for next generation
  wireless systems,'' \emph{IEEE Wireless Commun.}, vol.~22, no.~5, pp.
  98--104, Oct. 2015.

\bibitem{3GPP}
{TR 36.814 V1.2.0}, ``Further advancements for {E-UTRA} physical layer aspects
  (release 9),'' \emph{3GPP}, Jun. 2009.

\bibitem{jeffrey2007table}
A.~Jeffrey and D.~Zwillinger, \emph{Table of integrals, series, and products},
  6th~ed.\hskip 1em plus 0.5em minus 0.4em\relax Academic Press, 2000.

\end{thebibliography}
\end{document}